\documentclass[aps,pra,reprint]{revtex4-1}

\usepackage{graphicx}
\usepackage{amsmath,amssymb,amsthm}

\newtheorem{thm}{Theorem}

\newtheorem{coro}{Corollary}
\newtheorem{prop}{Proposition}
\newtheorem{lemma}{Lemma}
\theoremstyle{remark}
\newtheorem{ex}{Example}
\newtheorem{rem}{Remark}
\def\<{\langle}
\def\>{\rangle}
\def\states{\mathfrak S}

\def\mea{\mathsf m}

\def\simp{\mathsf s}
\def\Simp{\mathsf S}
\def\edg{\mathsf e}
\def\Tr{\mathrm {Tr}\,}
\def\ptr{\mathrm {Tr}}
\def\Ha{\mathcal H}

\def\Ce{\mathcal C}
\def\aff{\mathcal A}

\begin{document}
\title{Incompatible measurements in a class of general probabilistic theories}

\author{Anna Jen\v cov\'a}
\email[]{jenca@mat.savba.sk}
\affiliation{Mathematical Institute, Slovak Academy of Sciences, \v Stef\'anikova 49, 814 73 Bratislava, Slovakia}

\begin{abstract}
We study  incompatibility of measurements and its relation to steering and nonlocality in a class of finite dimensional general probabilistic theories (GPT). The basic idea is to represent finite collections of measurements as affine maps of a state space into a polysimplex 
and show that incompatibility  is characterized by properties of these maps. We introduce the notion of an incompatibility witness and 
 show its relation to incompatibility degree.  We find the maximum incompatibility degree attainable by pairs of two-outcome quantum  measurements
 and  characterize state spaces for which incompatibility degree attains maximal values possible in GPT.
 As examples, we study the spaces of classical and quantum channels and show their close relation to 
  polysimplices. This relation explains the super-quantum non-classical effects that were observed on these spaces.
\end{abstract}

\maketitle

\section{Introduction}
Incompatibility of measurements is one of the fundamental  features of quantum mechanics. As a key ingredient in quantum information protocols, incompatibility and related non-classical effects such as Bell nonlocality and steering, are viewed as important resources in quantum information theory.

General probabilistic theories (GPT) form a framework for description of  physical theories admitting probabilistic processes. 
Within this framework, the  quantum theory is specified by several axioms  \cite{hardy2001quantum,cdp2011informational}. 
Some of these axioms (causality, prefect distinguishability and local tomography) define a class of theories that can be studied within the setting of ordered vector spaces and their tensor products.

Many of the basic features that distinguish quantum mechanics  from any classical theory are shared by a large class of GPT. 
It is a natural question what are the properties that characterize quantum mechanics. This question has been studied for many years, see e.g. \cite{araki1980onacharacterization, bineu1936thelogic, mackey1963themathematical}, but recently the advance of quantum information theory led to a renewed interest in this topic,  \cite{barrett2007information, buste2014steering, bblw2006cloning, bblw2007generalized}. To answer this question, it is important to understand the nature of the non-classical features and relations between them. 

Such relations were already observed: it is well known that steering and nonlocality require both entanglement and incompatible observables.  The relations between non-locality, steering and incompatibility were studied in \cite{wopefe2009,buste2014steering,bggk2013degree, quvebr2014joint, banik2015measurement}, both for  quantum theory and 
in GPT. On the other hand, there exist unsteerable entangled quantum states 
\cite{wjd2007steering}, and incompatible sets of quantum measurements  that cannot lead to violation of Bell inequalities \cite{quvebr2014joint}.  To understand these relations, the more general setting is useful because it allows one to 
recognize which  non-classical manifestations  are consequences of convexity and the tensor product structure, and which are inherently quantum.

In this contribution, we study incompatibility of measurements and its relation to steering and nonlocality in a class of finite dimensional GPT, using the tools of convex geometry.
The basic idea is to represent finite collections of measurements as affine maps of a state space into a polysimplex 
(that is, a Cartesian product of simplices) and show that 
incompatibility  is characterized by properties of these maps.  A generalization of this idea was already used to 
 describe incompatibility of channels \cite{plavala2017conditions}. We introduce the notion of an incompatibility witness and 
 show its relation to incompatibility degree, defined in \cite{hstz2014maximally, hmz2016aninvitation}. We find the maximum incompatibility degree attainable by pairs of two-outcome quantum  measurements, generalizing the results of \cite{bggk2013degree},  and  characterize state spaces for which incompatibility degree attains maximal values possible in GPT. This completes the results of  
 \cite{bhss2013comparing} and \cite{jepl2017conditions}, where maximal incompatibility of  pairs of two-outcome measurements is considered.  
 The representation of collections of measurements as maps enables us to  tie incompatibility directly to steering and non-locality of states of composite systems. The concept 
  of incompatibility witnesses helps to explain the relation of incompatibility degree and maximal violation of Bell inequalities, as well as the observed limitations of these relations.
 
Besides the classical and quantum state spaces, we study the spaces of classical and quantum channels. It was observed in \cite{jepl2017conditions, srcz2016incompatible} that these spaces admit maximally incompatible measurements, which is known to be impossible in finite dimensional quantum theory  \cite{hstz2014maximally}.  Moreover, it was shown that causal bipartite 
quantum channels can be used to obtain  maximal violation of the CHSH inequality \cite{bgnp2001causal,plzi2017popescu}, in fact, all kinds of non-signalling correlations \cite{hosa2017achannel}. 
We prove these results as easy consequences of the fact that the set of classical channels is a retract of the set of quantum channels
 and is affinely isomorphic to a polysimplex.  

The paper is organized as follows. In the next section we briefly describe the main components of GPT in our setting and present the 
main examples of spaces of classical and quantum states and channels. The structure of polysimplices is detailed in Section \ref{sec:poly}. The relations to spaces of channels are also proved. Section \ref{sec:inc} is devoted to incompatibility of measurements, incompatibility witnesses and degree.  The last section deals with steering and Bell nonlocality.

\section{General probabilistic theories}

We  present a brief overview of  GPT in the  finite dimensional setting, explain our overall assumptions and introduce the mathematical tools needed in the sequel. Let us remark that the GPT framework is much broader, 
 see e.g. \cite{barrett2007information, chdp2010probabilistic, jahi2014generalized, bawi2012postclassical} for more details. 

We first recall a few definitions and  facts about convex sets and ordered vector spaces that will be needed below. 
For a subset $X\subseteq V$ of a finite dimensional vector space $V$, we 
 denote by $co(X)$ the convex hull and  $aff(X)$ the affine span of $X$ in $V$. For a convex subset $C\subseteq V$,  $ri(C)$ denotes the relative interior of $C$ in $aff(C)$ and $\dim(C):=\dim(aff(V))$. The set of all affine maps (that is, preserving the convex structure) between convex sets $C_1$ and $C_2$ will be denoted by $\aff(C_1,C_2)$. 

For the purposes of this paper, an ordered vector space is a pair $(V,V^+)$, where $V$ is a real finite dimensional  vector space and $V^+$ is a closed convex cone satisfying $V^+\cap -V^+=\{0\}$ and $V=V^+-V^+$. This  induces a partial order in $V$ as $v\le w$ if $w-v\in V^+$.
Let $V^*$ be the vector space dual with duality $\<\cdot,\cdot\>$. The order dual of $(V,V^+)$ is the ordered vector space  $(V^*,(V^+)^*)$, with the closed convex cone of positive functionals
\[
(V^+)^*:=\{v^*\in V^*,\ \<v^*,v\>\ge 0,\ \forall v\in V^+\}.
\]
Note that we have $(V^+)^{**}=V^+$. A linear map between ordered vector spaces is called positive if it preserves the positive cone.
We say that the cone $V^+$ is weakly self-dual if it is affinely isomorphic to $(V^+)^*$.

Let $(V,V^+)$ and $(W,W^+)$ be ordered vector spaces. There are two  distinguished ways to define a 
positive cone in the  tensor product $V\otimes W$:
\begin{align*}
V^+\otimes_{min} W^+ &:= \{\sum_i v_i\otimes w_i,\ v_i\in V^+, w_i\in W^+\}\\
V^+\otimes_{max} W^+ &:= ((V^+)^*\otimes_{min} (W^+)^*)^*.
\end{align*}
We have $V^+\otimes_{min} W^+\subseteq V^+\otimes_{max} W^+$. The elements of $V^+\otimes_{min} W^+$ are called separable. 

\subsection{States, effects and measurements}

 The framework of GPT is build on  basic notions of  states, representing  preparation procedures of a given system, and effects, assigning to each  state the corresponding probabilities of outcomes in yes/no experiments. The  state spaces  have  a  natural convex structure, expressing the possibility of forming probabilistic mixtures of states. 
The effects must respect this structure and are therefore represented by affine functions from the state space into the unit interval. 
Throughout this paper, we will assume that any state space is a compact convex subset $K$ of a finite dimensional real vector space. 
Moreover, we adopt the no restriction hypothesis, requiring that all affine functions $K\to [0,1]$ correspond to physical effects. 
 For a discussion of these assumptions in GPT
see \cite{chdp2010probabilistic, cdp2011informational, jala2013generalized}.

A state space $K$ determines a pair of dual ordered vector spaces as follows. Let $A(K):= \aff(K,\mathbb R)$ and $A(K)^+:=\aff(K,\mathbb R^+)$. Let also $1_K$ be the constant map $1_K(x)\equiv 1$. Then $(A(K), A(K)^+)$ is an ordered vector space. The function $1_K$ is an interior element in $A(K)^+$ and hence is an order unit: for any $f\in A(K)$ we have $-t1_K\le f\le t1_K$ for some $t>0$.
The set of effects is thus given by
\[
E(K):=\{f\in A(K),\ 0\le f \le 1_K\}.
\]

For $x\in K$, the evaluation map $f\mapsto f(x)$ defines a linear functional on $A(K)$ that is clearly positive and unital: $1_K(x)=1$. The converse is also true \cite{alfsen1971compact}, so that $K$ can be identified with the set of  positive unital functionals, or states, on $A(K)$.
With this identification, $K$ is a base of the dual cone $(A(K)^+)^*$, in the sense that each $0\ne \varphi\in(A(K)^+)^* $  can be expressed in a unique way as a multiple of some element in $K$. We therefore have $(A(K)^+)^*\equiv V(K)^+:=\cup_{\lambda\ge 0}\lambda K$ and 
$A(K)^*\equiv V(K):=V(K)^+-V(K)^+$. These identifications will be used throughout. 
  For $\psi\in V(K)$, let 
\[
\|\psi\|_K=\inf\{ a+b,\ \psi=ax-by,\ a,b\ge 0, x,y\in K\}.
\]
Then $\|\cdot\|_K$ is a norm in $V(K)$, called the base norm. It is the dual of the supremum norm $\|\cdot\|_{max}$ in $A(K)$. 

Similarly as for two-outcome measurements, any  measurement on a system with state space $K$ is fully described by its outcome statistics in each state. A measurement with $n+1$ outcomes is therefore identified with a map $f\in \aff(K,\Delta_n)$,  where $\Delta_n$ is the $n$-dimensional simplex  of probability measures over $\{0,\dots,n\}$.  The measurement $f$ is determined by $n+1$ effects $f_0,\dots,f_n\in E(K)$, satisfying $\sum_i f_i=1_K$, given as  $f_i(x)= f(x)(i)$. As before, we assume that each element of $\aff(K,\Delta_n)$ describes a valid measurement.

\begin{ex}\label{ex:classical}(\textbf{Classical state spaces.})  The state space of a classical system is an $m$-dimensional simplex 
$\Delta_m$. We have $A(\Delta_m)\simeq V(\Delta_m)=\mathbb R^{m+1}$, with $V(\Delta_m)^+\simeq A(\Delta_m)^+$ the positive cone of vectors with nonnegative entries and $E(\Delta_m)$ is the set of vectors with entries in $[0,1]$. The base norm in this case is 
 the $l_1$-norm in $\mathbb R^{m+1}$. Measurements  
$f\in \aff(\Delta_m,\Delta_n)$ are classical channels and can be identified with $(m+1)\times (n+1)$ stochastic matrices $\{T(j|i)\}_{i,j}$, where  $T(\cdot|i)\in \Delta_n$, $i=0,\dots,m$ are  determined by the values of $f$ on the vertices of $\Delta_m$.

\end{ex}

\begin{ex}\label{ex:quantum}(\textbf{Quantum state spaces.}) A quantum state space is the set of density operators $\states(\Ha)$ on a  finite dimensional Hilbert space $\Ha$. We will sometimes use labels  $\Ha_A$, $\Ha_B$, etc. for the Hilbert spaces, then we use the notations $d_A:=\dim(\Ha_A)$, $\states_A:=\states(\Ha_A)$. 
Any $f\in A(\states(\Ha))$ has the form
\[
f(\rho)=\Tr M\rho,\qquad \rho\in \states(\Ha)
\]
for some $M\in B_h(\Ha)$, the set of Hermitian operators  on $\Ha$. In this way, we have
$A(\states(\Ha))\simeq V(\states(\Ha))=B_h(\Ha)$,  $A(\states(\Ha))^+\simeq V(\states(\Ha))^+=B(\Ha)^+$, the cone of positive operators on $\Ha$, $1_{\states(\Ha)}=I$, the identity operator and $E(\states(\Ha))\simeq E(\Ha)$, the set of quantum effects. The base norm $\|\cdot\|_{\states(\Ha)}$ is the trace norm $\|X\|_1= \Tr|X|$. 
The measurements are given by positive operator valued measures (POVMs) on $\Ha$, that is, tuples of effects $M_0,\dots, M_n\in B(\Ha)^+$, $\sum_i M_i=I$. 

\end{ex}

\begin{ex}\label{ex:qchannels}(\textbf{Spaces of quantum channels.})
Let $\mathcal H_A,\mathcal H_{A'}$ be finite dimensional Hilbert spaces. We will denote by $\mathcal C_{A,A'}$ the set of all quantum 
channels $\mathcal H_A\to \mathcal H_{A'}$, that is, all completely positive and trace preserving linear maps $B(\mathcal H_A)\to 
B(\mathcal H_{A'})$. We now describe the corresponding cones and measurements for $\mathcal C_{A,A'}$, see \cite{jencova2012generalized} for details.

By the Choi representation, $\mathcal C_{A,A'}$ is isomorphic to a compact convex subset of the quantum state space $\states_{A'A}$.  Using this isomorphism, 
$V(\mathcal C_{A,A'})$ can be identified with the subspace 
\[
V(\Ce_{A,A'})\equiv \{X\in B_h(\mathcal H_{A'A}), \ \ptr_{A'} X \in \mathbb R I\}, 
\]
where $\ptr_{A'}$ is the partial trace over $\Ha_{A'}$. We then have
\begin{align*}
\mathcal C_{A,A'}&=V(\mathcal C_{A,A'})\cap \mathfrak S_{A'A}.
\end{align*}
Consequently, $A(\Ce_{A,A'})$ is a quotient of $B_h(\mathcal H_{A'A})$ and $A(\Ce_{A,A'})^+$ is the 
set of equivalence classes containing some positive element. The base norm $\|\cdot\|_{\Ce_{A,A'}}$ is identified with 
the diamond norm $\|\cdot\|_\diamond$ \cite{jencova2014base}.
One can also show that any measurement  $f\in \aff(\Ce_{A,A'},\Delta_n)$ has the form 
\[
f_i(\Phi)= \Tr M_i(\Phi\otimes id_{R})(\rho_{AR})
\]
for some POVM $M_0,\dots, M_n$ on $\mathcal H_{A'R}$  
and some state $\rho\in \mathfrak S_{AR}$ where  $\mathcal H_R$ is an ancilla, $d_R\le d_A$, but the representation in this form is not unique, see also \cite{ziman2008process}. In particular, the unit effect $1_{\Ce_{A,A'}}$ is obtained from any state 
$\rho_{AR}$ and the trivial measurement $M_0=I_{A'R}$. 
 
\end{ex}

\begin{ex}\label{ex:CC_QC} (\textbf{Spaces of classical channels.}) The set of classical channels $\aff(\Delta_m,\Delta_n)$  is isomorphic to a subset of $\Ce_{A,A'}$, 
with $d_A=m+1$, $d_{A'}=n+1$.  Such an isomorphism is obtained by fixing orthonormal bases $|i\>_A$ of $\Ha_A$ and $|j\>_{A'}$ of $\Ha_{A'}$ and putting for any stochastic matrix $T\in \aff(\Delta_m,\Delta_n)$,
\begin{equation}\label{eq:CC_T}
\Phi_T(\sigma)=\sum_{i,j} \<i,\sigma |i\>_A T(j|i) |j\>\<j|_{A'},\qquad \sigma\in \states(\Ha).
\end{equation}
Quantum channels of this form are  called classical-to-classical, or c-c channels. The cones and measurements for this state space will be identified later (cf. Proposition \ref{prop:ccchannel_polysimp}). 

\end{ex}

\subsection{Composition of state spaces: tensor products}

Let $K_A$ and $K_B$ be  state spaces, corresponding to two systems in a GPT. To describe the state space of the joint system, we need the notion of a tensor product of state spaces. Let  the composite state space be denoted by $K_A\widetilde\otimes K_B$. Assuming the local tomography axiom \cite{hardy2001quantum, chdp2010probabilistic,cdp2011informational},  $K_A\widetilde\otimes K_B$ is a subset of the tensor product $V(K_A)\otimes V(K_B)$  such that
\begin{enumerate}
\item[(a)]  $x_A\otimes x_B\in  K_A\widetilde\otimes K_B$ for all $x_A\in K_A$, $x_B\in K_B$,
\item[(b)] $f_A\otimes f_B\in E(K_A\widetilde\otimes K_B)$ for all $f_A\in E(K_A), f_B\in E(K_B)$ 
\item[(c)] $1_{K_A\widetilde\otimes K_B}= 1_{A}\otimes 1_{B}$, here $1_A:=1_{K_A}$, $1_B:=1_{K_B}$.
\end{enumerate}
This is based on the requirement  that for the composite system, all product states  and all product effects are valid.
These conditions determine the minimal and the  maximal tensor product of state spaces. Let 
\begin{align*}
K_A\otimes K_B&:=co\{x_i\otimes y_i,\ x_i\in K_A,\ y_i\in K_B\}\\
K_A\widehat\otimes K_B&:=\{y\in V(K_A)\otimes V(K_B), \<f_A\otimes f_B,y\>\ge0,\\
                     & \forall f_A\in E(K_A), f_B\in E(K_B), \ \<1_{A}\otimes 1_{B},y\>=1\}
\end{align*}
Note that  both sets satisfy the conditions for a composite state space and we always have
\[
K_A\otimes K_B\subseteq K_A\widetilde\otimes K_B\subseteq K_A\widehat\otimes K_B.
\]
The states in $K_A\otimes K_B$ are called separable, all other states in $K_A\widetilde\otimes K_B$ are called entangled. 
 The particular form of the composite state space is specified by the theory in question, see the examples below.
 In terms of the related spaces and cones, we have
\begin{align*} 
V(K_A\otimes K_B)&\simeq V(K_A\widehat\otimes K_B)\simeq V(K_A)\otimes V(K_B),\\
A(K_A\otimes K_B)&\simeq A(K_A\widehat\otimes K_B)\simeq A(K_A)\otimes A(K_B),\\
V(K_A\otimes K_B)^+&\simeq V(K_A)^+\otimes_{min} V(K_B)^+,
\\ V(K_A\widehat\otimes K_B)^+&\simeq V(K_A)^+\otimes_{max} V(K_B)^+,
\\A(K_A\otimes K_B)^+&\simeq A(K_A)^+\otimes_{max} A(K_B)^+,
\\ A(K_A\widehat\otimes K_B)^+&\simeq A(K_A)^+\otimes_{min} A(K_B)^+,
\end{align*}
this follows easily from the definitions and duality relations.

\begin{ex}\label{ex:composite_class} For classical state spaces, we have
\[
\Delta_{n_A}\otimes \Delta_{n_B}=\Delta_{n_A}\widehat\otimes \Delta_{n_B}=\Delta_{n_{AB}},
\]
$n_{AB}:=n_An_B+n_A+n_B$, is the set of probability measures on $\{0,\dots,n_A\}\times\{0,\dots,n_B\}$.  
In fact,  we have $K\otimes \Delta_n= K\widehat \otimes \Delta_n$ for any state space  $K$ and this property characterizes the simplices in a general infinite-dimensional setting, see \cite{naph1969tensor}. 

\end{ex}

\begin{ex}\label{ex:composite_quant} For quantum state spaces, we have
\[
\states_{A}\widetilde \otimes \states_{B}=\states_{AB},
\]
with the usual tensor product of Hilbert spaces $\Ha_{AB}=\Ha_A\otimes\Ha_B$. Here the minimal tensor product 
$\states_{A}\otimes \states_{B}$ is the subset of separable states and the maximal tensor product 
$\states_{A}\widehat \otimes \states_{B}$ is the set of entanglement witnesses (with unit trace). 
\end{ex}

\begin{ex}\label{ex:composite_QC} For spaces of quantum channels, let  $\Phi\in \Ce_{A,A'}\widetilde \otimes \Ce_{B,B'}$.
It is natural to require that $\Phi\in \Ce_{AB,A'B'}$, so that $\Phi$  is a bipartite quantum channel. 
By the condition (b), each product of effects is a valid effect, in particular, $1_{\Ce_{A,A'}}\otimes f_B$ is a valid effect for any
 $f_B\in E(\Ce_{B,B'})$.
By Example \ref{ex:qchannels}, $1_{\Ce_{A,A'}}$ is obtained from any state $\rho_{RA}$ and the identity $I_{RA'}$. Let $f_B$ be given by $\sigma_B$ and an effect $M_{B'}$, then 
\[
\<1_{\Ce_{A,A'}}\otimes f_B,\Phi\>= \Tr (I_{RA}\otimes M_{B'})(id_R\otimes \Phi)(\rho_{RA}\otimes \sigma_B)
\]
and this expression does not depend on $\rho_{RA}$.
It follows  that  
\[
\sigma_{B}\mapsto \ptr_{RA}(id_{R}\otimes \Phi)(\rho_{RA}\otimes \sigma_{B})
\]
defines a channel in $\Ce_{B,B'}$ that does not depend on $\rho$; similarly, we obtain a channel in $\Ce_{A,A'}$ 
 by applying the unit effect on the second part. Channels with this property are called causal or no-signalling bipartite channels, see \cite{bgnp2001causal}, the set of all such channels is denoted by $\Ce^{caus}_{AB,AB'}$. 
 We define the composite state space as
 \[
\Ce_{A,A'}\widetilde \otimes \Ce_{B,B'}:=\Ce^{caus}_{AB,A'B'}.
 \]
 The minimal tensor product is the set of local bipartite channels, which are convex combinations of 
channels prepared by each party separately. The maximal tensor product is strictly larger than $\Ce^{caus}_{AB,A'B'}$, 
since its elements are not necessarily completely positive.

\end{ex}

\begin{ex}\label{ex:composite_CC} 
 It is clear that 
\[
\aff(\Delta_{m_A},\Delta_{n_A})\widehat \otimes \aff(\Delta_{m_B},\Delta_{n_B})\subset \aff(\Delta_{m_{AB}},\Delta_{n_{AB}}),
\]
where $\Delta_{m_{AB}}=\Delta_{m_Am_B+m_A+m_B}=\Delta_{m_A}\otimes \Delta_{m_B}$ and similarly for $\Delta_{n_{AB}}$.
We will see later (Section \ref{sec:nonlocQC}) that the maximal tensor product is the set of all classical bipartite causal 
channels, characterized by the no-signalling conditions \eqref{eq:no_signallingA} and \eqref{eq:no_signallingB}.

\end{ex}

\subsection{Channels and positive maps}

Channels in GPT describe transformations of the systems allowed in the theory and are represented by affine maps between state spaces. 
Although all affine maps between simplices are classical channels, we do not assume in general that all elements in $\aff(K,K')$ for state spaces $K$ and $K'$ correspond to valid channels. 
For the spaces of quantum states and channels, it is required that the maps have completely positive extensions.
Completely positive maps $B(\Ha_{A'A})\to B(\Ha_{B'B})$ that map $\Ce_{AA'}$ into $\Ce_{BB'}$ are called  quantum supermaps \cite{cdp2008supermaps}
 or quantum combs \cite{cdp2009framework} and belong to a hierarchy describing quantum networks.

Any  $T\in \aff(K,V(K'))$ extends uniquely to a linear map 
 $T: V(K)\to V(K')$ and  $\aff(K,V(K'))$ has the structure of a real vector space. The subset 
 $\aff(K,V(K')^+)\subseteq\aff(K,V(K')) $ is a closed convex cone of elements whose extensions are positive maps. 
 With this cone,   $\aff(K,V(K'))$ is an ordered vector space.

Let $T_A\in \aff(K_A,V(K_A')^+)$ and $T_B\in \aff(K_B,V(K_B')^+)$. It is easy to see that $T_A\otimes T_B$ is positive 
 with respect to both the maximal and minimal tensor product cones, that is,
\[
T_A\otimes T_B\in \aff(K_A\otimes K_B, 
V(K_A'\otimes K_B')^+)
\]
and  
\[
T_A\otimes T_B\in \aff(K_A\widehat \otimes K_B, V(K_A'\widehat \otimes K_B')^+).
\]
Indeed, the first inclusion  follows from the definition of the minimal tensor product and the second one from
\[
\<(T_A\otimes T_B)(y), f_A'\otimes f_B'\>=\<y, T_A^*(f_A')\otimes (T_B)^*(f_B')\>,
\]
for all $y\in K_A\widehat\otimes K_B$, $f_A'\in A(K_A')^+$ and $f_B'\in A(K_B')^+$, here $T^*$ denotes the adjoint of the linear extension of $T$. We say that $T_A$ is entanglement breaking (ETB) if for any state space $K_B$, we have 
$T_A\otimes id_{K_B}\in \aff(K_A\widehat \otimes K_B, V(K_A'\otimes K_B)^+)$. The set of all ETB maps will be denoted by 
$\aff_{sep}(K_A, V(K_A')^+)$, it is  a closed convex subcone in   $\aff(K_A, V(K_A')^+)$. 

There is a well known relation between linear maps and tensor products of vector spaces, with respect to which the positive maps
 correspond to elements of the maximal tensor product and ETB maps to elements of the minimal one. Details on these relations, as well as the proofs of the following results, are given in Appendix \ref{app:positive}.

\begin{prop}\label{prop:ETB} Let $T\in \aff(K,V(K')^+)$. Then 
 $T$ is ETB if and only if $T$  factorizes through a simplex: there are a simplex $\Delta_n$ and maps $T_0\in \aff(K,V(\Delta_n)^+)$ and $T_1\in \aff(\Delta_n, V(K')^+)$ such that $T=T_1T_0$. 
If $T$ is a channel, $T_0$ and $T_1$   may be chosen to be channels as well.
\end{prop}

It is clear that any constant map factorizes through the 1-dimensional simplex $\Delta_0$ and hence must be ETB.

We now look at the dual spaces and cones. For $T\in \aff(K,V(K))$, let $\Tr T$ denote the usual trace of its linear extension.
It is not difficult to see that 
the dual space of $\aff(K,V(K'))$ can be identified with $\aff(K',V(K))$, with duality $\<S,T\>=\Tr ST$.

\begin{prop}\label{prop:aff_dual}
The dual cone to $\aff(K,V(K')^+)$ is $\aff_{sep}(K',V(K)^+)$. 

\end{prop}

 \section{Polysimplices and their structure}\label{sec:poly}
Let $k,l_0,\dots,l_k\in \mathbb N$. A polysimplex is a Cartesian product of simplices
\[
\Simp_{l_0,\dots,l_k}:=\Delta_{l_0}\times\dots\times \Delta_{l_k},
\]
with pointwise defined convex structure. This is a compact convex set, more precisely  a convex polytope. Elements of  $\Simp_{l_0,\dots,l_k}$ represent 
states of a device determined by a set of inputs indexed by $0,\dots,k$, each of which has an allowed set of outputs $0,\dots,l_i$. 
Such devices were introduced by Popescu and Rohrlich \cite{poro1994quantum} as toy theories exhibiting super-quantum correlations.
In the framework of GPT, theories with state spaces  of this form  were  studied in \cite{barrett2007information, jala2013generalized}.

If $l_0=\dots=l_k=n$, the polysimplex  will be denoted by  $\Delta_n^{k+1}$. The $(k+1)$-hypercube $\Delta_1^{k+1}$ will be denoted by $\square_{k+1}$. If $l_0,\dots,l_k\in \mathbb N$ are   assumed fixed, we often drop the multiindex $l_0,\dots,l_k$ and denote the polysimplex by $\Simp$.

Let $f^0,\dots, f^k$ be a collection of measurements on $K$, such that $f^i\in \aff(K,\Delta_{l_i})$. By definition of the Cartesian product, 
 such collections correspond precisely to elements of $\aff(K,\Simp_{l_0,\dots,l_k})$. 
Explicitly, the relations between $f^i\in \aff(K,\Delta_{l_i})$, $i=0,\dots,k$ and $F=(f^0,\dots,f^k)\in  \aff(K,\Simp_{l_0,\dots,l_k})$ are given by
\begin{equation}\label{eq:fF}
F(x)=(f^0(x),\dots, f^k(x));\quad f^i=\mea^i F,\ \forall i 
\end{equation}
where $\mea^i:\Simp_{l_0,\dots,l_k}\to \Delta_{l_i}$ is the projection onto the $i$-th component.
Since our main results are based on the relation
  of properties of such maps to incompatibility, it will be necessary to describe the structure of the 
  polysimplices and related spaces and cones. 

The vertices of $\Simp=\Simp_{l_0,\dots,l_k}$ are  the $(k+1)$-tuples
\[
\simp_{n_0,\dots,n_k}:=(\delta^0_{n_0},\dots,\delta^k_{n_k}),\quad n_i=0,\dots,l_i,\ i=0\dots,k,
\]
where $\delta^i_{n_i}$ denotes the $n_i$-th vertex of the $i$-th simplex. If $u_i\in \Delta_{l_i}$ is the uniform probability distribution for all $i$, then
\[
\bar\simp:=(u_0,\dots,u_k)=\frac1{\Pi_i(l_i+1)}\sum_{n_0,\dots,n_k} \simp_{n_0,\dots,n_k}
\]
is the barycenter of $\Simp$.  Further, note that each projection $\mea^i$ is a measurement on 
$\Simp$, with effects determined by 
\[
\mea^i_j(\simp_{n_0,\dots,n_k})= \left\{ \begin{array}{cc} 1 & \mbox{ if } n_i=j\\
   0 & \mbox{ otherwise}
   \end{array}\right..
\]
Since all faces of $\Simp$ have the form $F_0\times\dots\times F_k$, where $F_i$ is a face of $\Delta_{l_i}$,
it is clear that the maximal faces are precisely the null spaces of $\mea^i_j$. 

It will be convenient to fix a pair of dual bases of the spaces $A(\Simp)$ and $V(\Simp)$, such that the basis of $A(\Simp)$ consists
 of the effects $1_\Simp$ and $\mea^i_j$. Since $\sum_j\mea^i_j=1_\Simp$ for all $i$, we will  fix a linearly independent subset. 
For the dual basis, we need to describe the edges of $\Simp$. Since the edges  are 1-dimensional faces, they have the form 
\[
\{\delta^0_{n_0}\}\times\dots\times \{\delta^{i-1}_{n_{i-1}}\}\times E_i\times \{\delta^{i+1}_{n_{i+1}}\}\times\dots\times\{\delta^k_{n_k}\},
\]
where $E_i$ is an edge of $\Delta_{l_i}$. We see that  the vertices adjacent to a vertex $\simp_{n_0,\dots,n_k}$ are those that differ
 from $\simp_{n_0,\dots,n_k}$ in exactly one index, that is, the vertices
\[
\simp_{n_0,\dots,n_{i-1},j,n_{i+1},\dots,n_k},\ j\ne n_i,i=0,\dots,k. 
\]
Pick the vertex $\simp_{l_0,\dots,l_k}$ and let  
\[
\edg^i_j:= \simp_{l_0,\dots,l_{i-1},j,l_{i+1},\dots,l_k}-\simp_{l_0,\dots,l_k}
\]
denote the vectors given by the adjacent edges.

\begin{lemma}
\begin{enumerate}
\item[(i)] The extreme rays of the cone $A(\Simp)^+$ are generated by the effects $\mea^i_j$, $i=0,\dots,k$, $j=0,\dots,l_i$.
\item[(ii)] The effects
\begin{equation}\label{eq:basis}
1_\Simp, \mea^0_0,\dots,\mea^0_{l_0-1},\mea^1_0,\dots,\mea^1_{l_1-1},\dots, \mea^k_0,\dots,\mea^k_{l_k-1}
\end{equation}
form a basis of the vector space $A(\Simp)$.
\item[(iii)] The elements
\begin{equation}\label{eq:dualbasis}
\simp_{l_0,\dots,l_k},\edg^0_0,\dots,\edg^0_{l_0-1},\edg^1_0,\dots,\edg^1_{l_1-1},\dots,\edg^k_0,\dots,\edg^k_{l_k-1}
\end{equation}
form a basis of the vector space $V(\Simp)$, dual to \eqref{eq:basis}. 
\end{enumerate}

\end{lemma}

\begin{proof}
Since the null spaces $(\mea^i_j)^{-1}(0)$ are exactly the maximal faces of $\Simp$, 
these effects generate the extreme rays of $A(\Simp)^+$.  
Further, let $f_1,f_2,\dots$ denote the elements in \eqref{eq:basis} and $x_1,x_2,\dots$ the elements of \eqref{eq:dualbasis}. It is easy to see that 
$\<f_i,x_j\>= \delta_{ij}$, so that both sets are linearly independent. The statements (ii) and (iii) now follow from the fact that 
$\dim(A(\Simp))=\sum_{i=0}^k l_i+1$.

\end{proof}
The basis elements \eqref{eq:dualbasis} are visualized in Fig. \ref{fig:basis}. 
With respect to this basis, the vertices are  expressed as
\begin{equation}\label{eq:vertices}
\simp_{n_0,\dots,n_k}=\simp_{l_0,\dots,l_k}+\sum_{i=0}^k \edg^i_{n_i}.
\end{equation}

\begin{rem}\label{rem:basis} We can get another pair of dual bases using any vertex $\simp_{n_0,\dots,n_k}$ and its adjacent edges
 for a basis of $V(\Simp)$ and  
\[
\{1_\Simp\}\cup \{\mea^i_j,\ j\ne n_i, i=0,\dots,k\}
\]
for the dual basis of $A(\Simp)$. 

\end{rem}
\begin{figure*}
\begin{minipage}[c]{0.5\textwidth}
\centering
\includegraphics{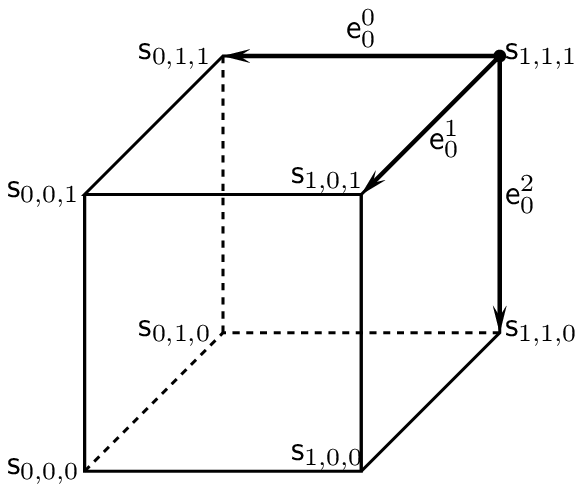}
\end{minipage}%
\hfill
\begin{minipage}[c]{0.5\textwidth}
\centering
\includegraphics{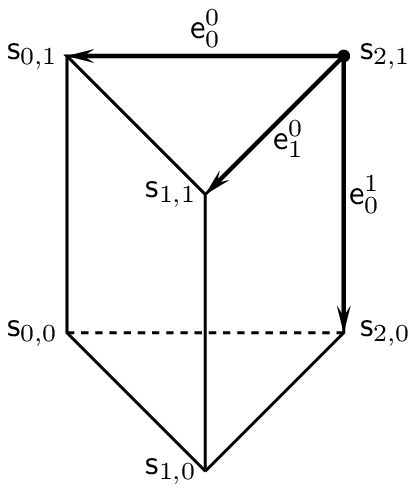}
\end{minipage}
\caption{Basis elements $\{\simp_{1,1,1},\edg^0_0,\edg^1_0,\edg^2_0\}$ for the cube $\square_3$, and $\{\simp_{2,1},\edg^0_0,\edg^0_1,\edg^1_0\}$ for the polysimplex $\Simp_{2,1}$.}\label{fig:basis}
\end{figure*}

\begin{ex}\label{ex:square}(\textbf{The square}.)
The simplest example  is the square $\square_2=\Delta_1\times \Delta_1$. The vertices  
$\simp_{0,0},\simp_{0,1},\simp_{1,0},\simp_{1,1}$  satisfy the relation $\simp_{0,0}+\simp_{1,1}=\simp_{0,1}+\simp_{1,0}$.
This state space is also called the gbit or square-bit, see \cite{poro1994quantum, barrett2007information, bhss2013comparing}.
The dual cone $A(\square_2)^+$ is generated by four effects $\mea^0_0,\mea^0_1,\mea^1_0,\mea^1_1$. Since 
we have $\mea^0_0+\mea^0_1=1_{\square_2}=\mea^1_0+\mea^1_1$, 
these effects again form a square, so that  $V(\square_2)^+$ is weakly self-dual. Note that the square is the only 
 polysimplex  with this property.
(This follows from the fact that the extreme rays of  $V(\Simp)^+$ are generated by vertices of $\Simp$, whereas the extreme rays of $A(\Simp)^+$ correspond to maximal faces of $\Simp$. Therefore, $V(\Simp)^+\simeq A(\Simp)^+$ implies  $\Pi_{i=0}^k (l_i+1)=\sum_{i=0}^k(l_i+1)$, this holds only for the square). 
 
\end{ex}

The following relation of polysimplices and spaces of classical channels is immediate (see also Examples \ref{ex:classical} and \ref{ex:CC_QC}). 
 \begin{prop}\label{prop:ccchannel_polysimp} We have  $\Delta_n^{k+1}\simeq \aff(\Delta_{k},\Delta_n)$. More generally, any polysimplex 
 $\Simp_{l_0,\dots,l_k}$ is isomorphic to a face of $\aff(\Delta_{k},\Delta_n)$, with $n\ge \max_i l_i$. This isomorphism is given by \[
 s\mapsto T_s\in \aff(\Delta_{k},\Delta_n),\quad T_s(j|i)=\left\{\begin{array}{cc} \mea_j^i(s), &
\mbox{ if } j\le l_i\\
0 & \mbox{otherwise}\end{array}\right..
 \]

 \end{prop}

There is also a relation of polysimplices and spaces of quantum channels. To describe this relation, we will need the following notion.

Let $K$ and $K'$ be state spaces. A map $R\in \aff(K,K')$ such that there is a map $S\in \aff(K',K)$ with $R S=id_{K'}$ is called a retraction. The   map $S:K'\to K$ is then 
called a section. For any retraction-section pair $(R,S)$, the map $P=S R$ is a projection on $K$ onto the range of $S$, that is, an affine idempotent map $K\to K$ such that $P(K)=S(K')$. Moreover, any map in $\aff(K',C)$ for a convex set $C$ has an extension to a map in $\aff(K,C)$. 

\begin{prop}\label{prop:qchannel_polysimp}  There exists a retraction-section pair 
$R\in \aff(C_{A,A'}, \Delta^{d_A}_{d_{A'}-1})$ and $S\in \aff(\Delta^{d_A}_{d_{A'}-1},C_{A,A'})$, determined by 
\[
\mea^i_jR(\Phi)=\<j,\Phi(|i\>\<i|_A)|j\>_{A'},\quad \forall i,j; \Phi\in \Ce_{A,A'}
\]
 and 
\[
S(s)=\sum_{i,j} \<i,\cdot|i\>_A\mea^i_j(s)|j\>\<j|_{A'},\quad  s\in \Delta^{d_A}_{d_{A'}-1}.
\]

\end{prop}

\begin{proof} Since $\mea^i_jR(\Phi)\ge 0$ for all $i,j$ and $\sum_j \mea^i_jR(\Phi)=\Tr \Phi(|i\>\<i|)=1$ for all $i$,
 $R$ is a well defined element in $\aff(C_{A,A'}, \Delta^{d_A}_{d_{A'}-1})$. For each $s$, $S(s)$ is a c-c channel and is therefore completely positive. It is quite clear that 
 \[
\mea^i_jRS(s)=\mea^i_j(s),\quad \forall i,j,
 \]
so that $RS=id$.

\end{proof}

\begin{rem}\label{rem:channels}
The above Proposition implies that there is a projection  of $\Ce_{A,A'}$ onto a set of c-c channels and that
any map in $\aff(\Delta^{d_A}_{d_{A'}-1},C)$ can be extended to a map in $\aff(\Ce_{A,A'},C)$, for any convex set $C$. The consequences of this fact will become clear later on.
\end{rem}

\section{Incompatibility of measurements}\label{sec:inc}

 Let $K$ be a state space and  
let $f^i\in \aff(K,\Delta_{l_i})$ be a measurement with values in $\{0,\dots,l_i\}$, $i=0,\dots,k$. We say that  $f^0,\dots,f^k$  are compatible if they are the marginals of a single joint measurement with values in $\{0,\dots,l_0\}\times\dots\times \{0,\dots,l_k\}$. For an exposition of incompatibility in our setting,  see \cite{hmz2016aninvitation}.

The joint measurement is described by a map $g\in \aff(K,\Delta_L)$ with  $L:=\Pi_i(l_i+1)-1$ (note that $\Delta_L\simeq \bigotimes_i \Delta_{l_i}$) and can be parametrized as
 \[
 g(x)=\sum_{n_i\in \{0,\dots,l_i\}} g_{n_0,\dots,n_k}(x) \delta_{n_0,\dots,n_k},
 \]
where $g_{n_0,\dots,n_k}\in E(K)$ and $\delta_{n_0,\dots,n_k}$ is the probability measure concentrated at $(n_0,\dots,n_k)$. We then have
 \begin{equation}\label{eq:comaptible}
 f^i_j=\sum_{n_0,\dots,n_{i-1},n_{i+1},\dots,n_k} g_{n_0,\dots,n_{i-1},j,n_{i+1},\dots,n_k}.
 \end{equation}
for the effects of $f^i$. Let $J\in \aff( \Delta_{L},\Simp)$ be determined by $J(\delta_{n_0,\dots,n_k})=\simp_{n_0,\dots,n_k}$. Then it is easy to see that \eqref{eq:comaptible} can be written as  
$f^i_j=\mea^i_j Jg$. In other words, $f^0,\dots,f^k$ is compatible if and only if the corresponding $F\in \aff(K,\Simp)$ satisfies
\begin{equation}\label{eq:compatibleF}
F=Jg,\qquad\mbox{for some } g\in \aff(K,\Delta_L).
\end{equation}
The following observation is simple but important.

 \begin{thm}\label{thm:incompatibility} The measurements $f^0,\dots,f^k$ are compatible if and only if the corresponding channel $F$ is  entanglement breaking.
 \end{thm}

\begin{proof}  If $f^0,\dots,f^k$ are compatible, then $F$ is ETB by \eqref{eq:compatibleF} and  Proposition \ref{prop:ETB}.
Conversely, let $F$ be ETB. Proposition \ref{prop:ETB} implies that there is some simplex $\Delta_n$ and 
 channels $g'\in \aff(K,\Delta_n)$, $T\in \aff(\Delta_n,\Simp)$ such that $F=Tg'$. The channel $T$ corresponds to a collection of 
 measurements $t^i:=\mea^i T\in \aff(\Delta_n,\Delta_{l_i})$. Since all  measurements on a simplex are compatible, there is some
 $h\in \aff(\Delta_n,\Delta_{L})$ such that $T=Jh$.
Putting $g=hg'$ finishes the proof.

\end{proof}

\begin{rem}\label{rem:incomp_in_combs}
The above characterization of incompatible measurements as non-ETB channels 
suggests that these channels should be admissible in the GPT in question, which also means that we need to include the polysimplices into the theory. For quantum theory this might seem strange, since the polysimplices are certainly not quantum state spaces. On the other hand, the retraction-section pairs of Proposition \ref{prop:qchannel_polysimp}
 allow us to include maps in $\aff(\states(\Ha),\Simp)$ into the larger setting of quantum networks. If $F\in \aff(\states(\Ha), \Delta_n^{k+1})$ is a collection of quantum measurements and $S$ is the section of Proposition \ref{prop:qchannel_polysimp}, then $SF$ is a map from states 
  into quantum channels. Using the Choi representation, one can see that this map is also completely positive, hence a quantum comb, \cite{cdp2009framework}.  Moreover, since $RS=id$, $SF$ is ETB iff $F$ is. One should be aware that ''entanglement breaking'' has a different meaning here than for usual cp maps: $F$ is compatible iff $(SF\otimes id)(\rho)$ is a local bipartite channel
for  any bipartite state $\rho$.

\end{rem}

\subsection{Incompatibility witnesses}

Let $F=(f^0,\dots,f^k)\in \aff(K,\Simp)$ be a collection of measurements. By Proposition \ref{prop:aff_dual}, $F$ is non-ETB if and only if 
 there is some $W\in \mathcal A(\Simp,V(K)^+)$ such that $\Tr F W<0$. Such a $W$  will be called an 
incompatibility witness. As can be seen from Proposition \ref{prop:ETB}, this notion has a close relation to entanglement witnesses.

Any $W\in \mathcal A(\Simp,V(K))$ is determined by the images of the vertices of $\Simp$. The elements
\[
w_{n_0,\dots,n_k}:=W(\simp_{n_0,\dots,n_k})
\]
 will be called the vertices of $W$ (although not all of these points must be vertices of the image  $W(\Simp)$). The map $W$ is positive if and only if  all its vertices are in $V(K)^+$.  The image of the barycenter of $\Simp$,
  $\bar w:=W(\bar \simp)$,  will be called the barycenter of $W$.  We say that $W$ is degenerate if $\dim(W(\Simp))<\dim(\Simp)$.
A description of  the cones $\mathcal A(\Simp,V(K)^+)$ and $\mathcal A_{sep}(\Simp,V(K)^+)$ can be found in Appendix \ref{app:cones}.

It is clear that an ETB map cannot be an incompatibility witness. As we shall see (Fig. \ref{fig:qubitex} below), not all non-ETB maps are witnesses.
For a characterization of witnesses, we will need the following notion. 
Let $W,\tilde W\in \aff(\Simp,V(K)^+)$. We say that $\tilde W$  is a translation of $W$ in the direction $v\in V(K)$ if $\tilde W=W+L_v$ where $L_v$ is the constant map $L_v(s)\equiv v$. Equivalently, the vertices of $\tilde W$ satisfy
$\tilde w_{n_0,\dots,n_k}=w_{n_0,\dots,n_k}+v$ for all $n_0,\dots,n_k$. If $v$ is such that $\<1_K,v\>=0$, we say that $\tilde W$ is a translation of $W$ along $K$.

\begin{thm}\label{thm:witness} A map  $W\in \aff(\Simp,V(K)^+)$ is an incompatibility witness if and only if no translation of $W$ along $K$ is 
ETB.

\end{thm}

\begin{proof} Note that  $\mathcal A(K,\Simp)$ is a compact convex subset of $\mathcal A(K,V(\Simp)^+)$. We need to describe the generated space and cone. For shorter notations, let us denote $\mathcal V:=V(\mathcal A(K,\Simp))$ and $\mathcal V^+:=V(\mathcal A(K,\Simp))^+$. We have
\[
\mathcal V= \{T\in \mathcal A(K,V(\Simp)), 1_\Simp  T\in \mathbb R 1_K\},
\]
 and $\mathcal V^+=\mathcal A(K,V(\Simp)^+)\cap \mathcal V$. Let $\mathcal V^\perp$ be the annihilator of $\mathcal V$  in the dual space $\aff(\Simp,V(K))$, then it is not difficult to see that 
\[
 \mathcal V^\perp= \{L_v,\ v\in V(K),\<1_K,v\>=0\}.
\]
Since  $int(\mathcal V^+)\ne \emptyset$ (for example, any constant map of $K$ onto $s\in ri(\Simp)$ is in $int(\mathcal V^+)$), Krein's theorem \cite{naimark1959normed} implies that
 any positive functional  on $(\mathcal V,\mathcal V^+)$  extends to an element in the dual cone $\aff(K,V(\Simp)^+)^*=\aff_{sep}(\Simp,V(K)^+)$.
If  $W\in \aff(\Simp,V(K)^+)$ is not a witness, then $F\mapsto \Tr FW$ extends to a positive functional on $(\mathcal V,\mathcal V^+)$, so that  there is some $\tilde W\in \aff_{sep}(\Simp, V(K)^+)$ such that 
\[
\Tr FW=\Tr F\tilde W,\qquad F\in \mathcal A(K,\Simp).
\]
Hence   $W-\tilde W\in \mathcal V^\perp$, so that  $\tilde W$ is an ETB   translation of $W$ along $K$. Conversely, assume that $W$ is a witness. Let $F\in \aff(K,\Simp)$
 be such that $\Tr FW<0$, then for any translation $\tilde W$ of $W$ along $K$, we have $\Tr F\tilde W=\Tr FW<0$. It follows that $\tilde W$ is a witness as well and cannot be ETB. 

\end{proof}

\begin{figure*}
\begin{minipage}[c]{\textwidth}
\centering
\includegraphics{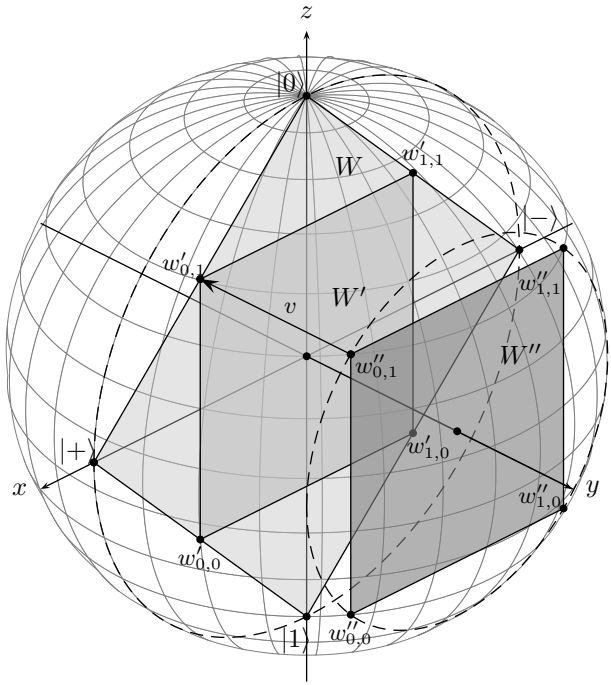}
\end{minipage}%
\caption{\footnotesize \textbf{Incompatibility witnesses for qubit states.} Three examples of maps $W, W', W''$ from the square into  the Bloch ball. The vertices of  $W$ are the pure states $|0\>$, $|1\>$, $|+\>$, $|-\>$ so that $W$ 
 is extremal and not ETB by Corollary \ref{coro:extr_nondeg}. It is easy to see that there is no nontrivial translation of $W$ along $K$,
  hence  $W$ is a witness. The map $W'$ is ETB, since the vertices $w_{i,j}'$ have a decomposition as in 
  Proposition \ref{prop:witness}, Appendix \ref{app:cones} (where  the elements $\psi^i_j$ are the vertices of $W$ multiplied by $\tfrac{1}{2}$). The map $W''$ has
   extremal vertices $w_{i,j}''$,  so that it is again extremal and not ETB by Corollary \ref{coro:extr_nondeg}. But $W''$ is 
   not a witness by Theorem \ref{thm:witness}, since the ETB map $W'$ is a translation of $W''$ along $K$.  }
\label{fig:qubitex}
\end{figure*}

We will find another characterization of incompatibility witnesses for two-outcome measurements later (Corollary \ref{coro:2witnesses}). 

\subsubsection{Extremal and non-ETB elements in $\aff(\square_2,V(K)^+)$}

For detection of incompatibility, it suffices to use witnesses that are extremal in the cone $\aff(\Simp,V(K)^+)$.
More precisely, for an ordered vector space $(V, V^+)$, we say that and element $v\in V^+$  is extremal  if it is nonzero and lies on an 
extreme ray of $V^+$. Alternatively, $v$ is extremal if $v\ne 0$ and $v'\le v$ for any  $v'\in V^+$ implies that $v'=tv$ for some $t\ge 0$. A description of extremal elements in $\aff(\Simp, V(K)^+)$ will be also useful in the next section. 

So far, we can do this in the simplest case when $\Simp=\square_2$.
 Any $W\in \aff(\square_2, V(K)^+)$ is given by four 
vertices  $w_{i,j}\in V(K)^+$, $i,j=0,1$, satisfying
\begin{equation}\label{eq:squaremap}
w_{0,0}+w_{1,1}=w_{0,1}+w_{1,0}=2\bar w.
\end{equation}

\begin{prop}\label{prop:extr} Let $W\in \mathcal A(\square_2,V(K)^+)$ have vertices  $w_{i,j}$, $i,j=0,1$ and barycenter $\bar w\ne 0$. Let 
$F_{i,j}$ denote the face of $V(K)^+$ generated by $w_{i,j}$ and let $L_{i,j}=F_{i,j}-F_{i,j}$ be the generated subspace. Then $W$ 
is extremal if and only if   
\begin{align}
L_{0,0}\cap L_{1,1}=L_{0,1}\cap L_{1,0}&=\{0\}\label{eq:extr_1}\\
(L_{0,0}\oplus L_{1,1})\cap (L_{0,1}\oplus L_{1,0})&=\mathbb R\bar w\label{eq:extr_2}.
\end{align}

\end{prop}

\begin{proof} 
Assume  that the two conditions are fulfilled  and let  $W'\le W$, with vertices $w_{i,j}'$ and barycenter $\bar w'$. Then clearly  $w_{i,j}'\in F_{i,j}$. 
By \eqref{eq:extr_2}  we must have $\bar w'=t\bar w$ for some $t\in [0,1]$, but then 
 $w'_{i,j}=tw_{i,j}$ for all $i,j$ by \eqref{eq:extr_1}. It follows that   $W$ is extremal.

Conversely, let us denote the subspace on the LHS of \eqref{eq:extr_2} by $L$ and assume that  there is some $\psi\neq t\bar w$ in $L$.  Then there are some $\psi_{i,j}\in L_{i,j}$ such that $\psi=\psi_{0,0}+\psi_{1,1}=\psi_{0,1}+\psi_{1,0}$. By definition of $L_{i,j}$, there is some $u>0$ such that $w_{i,j}^\pm:=
\tfrac12 w_{i,j}\pm u\psi_{i,j}\in V(K)^+$. Obviously, $w_{i,j}^+$ and $w_{i,1}^-$ are vertices of some  $W^+$ and $W^-$, which are 
not multiples of $W$, and we have  $W=W^++W^-$. It follows that  $W$ is not extremal. 

If \eqref{eq:extr_1} is not true, then there are some $\eta_{i,j}\in L_{i,j}$ such that not all of them are 0 and $\eta_{0,0}+\eta_{1,1}=\eta_{0,1}+\eta_{1,0}=0$.  We may then proceed as above to show that $W$ is not extremal. 

\end{proof}

We are now interested in extremal elements that are non-ETB. We start with a simple observation.

\begin{lemma}\label{lemma:degenerate} Let  $W\in \aff(\square_2,V(K)^+)$. If $W$ is degenerate, then it is ETB.
\end{lemma}

\begin{proof} $W$ is degenerate iff $\dim(W(\square_2))\le 1$. If the dimension is 0, then $W$ is constant, hence clearly ETB.
Assume that the dimension is 1, then  all vertices $w_{i,j}$ of $W$ lie on a segment. It is easy to see that we may find a decomposition as in 
Proposition \ref{prop:witness} using multiples of the endpoints of the segment, so that $W$ is ETB.

\end{proof}

\begin{coro}\label{coro:extr_nondeg}
Assume that $W$ is non-degenerate and each vertex is extremal in $V(K)^+$. Then $W$ is non-ETB and extremal in 
$\aff(\square_2,V(K)^+)$.
\end{coro}

\begin{proof} Any ETB map which is extremal in $\aff(\square_2,V(K)^+)$ has the form $\mea^i_j(\cdot)\phi$ for 
 some $i,j\in\{0,1\}$ and $\phi$ an extremal element in $V(K)^+$. Such a map  is clearly degenerate. It is therefore enough to show that $W$ is extremal.
Since $w_{i,j}$ are extremal in $V(K)^+$, $\dim(L_{i,j})=1$ for all $i,j$. If \eqref{eq:extr_1} or \eqref{eq:extr_2} is not satisfied, then it is easy to see that 
$\dim(W(\square_2)=1$. Since $W$ is non-degenerate, this is impossible.
 
\end{proof}

\begin{coro}\label{coro:2dim_witness} Let $\dim(K)=2$ and assume that  $W\in \mathcal A(\square_2,V(K)^+)$ is non-ETB.
Then $W$ is extremal if and only if all its vertices are extremal in $V(K)^+$. 

\end{coro}

\begin{proof} We will use the notation of Proposition \ref{prop:extr}. Assume that $W$ is extremal and that, say, $w_{0,1}=0$.
Then $w_{0,0}+w_{1,1}=w_{0,1}=2\bar w$ and $F_{0,0}, F_{1,1}\subseteq F_{0,1}$. By \eqref{eq:extr_2}, we must have
$L_{0,1}=\mathbb R\bar w$, so that $\phi$ must be extremal in $V(K)^+$. Consequently, both $w_{0,0}$ and $w_{1,1}$ are multiples of 
$\bar w$, but then $W$ is  degenerate and hence ETB.  It follows  that all vertices must be nonzero.  Then it follows from  \eqref{eq:extr_1}, \eqref{eq:extr_2} by dimension counting that  we must have $\dim(L_{i,j})=1$, so that all vertices are extremal in $V(K)^+$. The converse is Corollary \ref{coro:extr_nondeg}. 
\end{proof}

\begin{ex} (\textbf{The square})\label{ex:squarewitness} Since $\dim(\square_2)=2$, all non-ETB  extremal maps must have extremal vertices. Therefore,
$w_{0,0}$ and $w_{1,1}$ must be some multiples  of opposite  vertices, similarly $w_{0,1}$ and $w_{1,0}$ must be multiples 
of  the other pair of opposite vertices. Applying effects $\mea^i_j$ to the equality \eqref{eq:squaremap}, we see that all coefficients must be the same. It follows that $W$ is  (a multiple of) an automorphism of $\square_2$. Hence there are 8  extremal rays in $\mathcal A(\square_2,V(\square_2)^+)$ that are non-ETB. It is easy to see that elements in these rays are witnesses, since 
they have no nontrivial translations along $\square_2$.

\end{ex}

\begin{ex}(\textbf{Quantum state spaces}) \label{ex:quantumwitness} Let $W\in \mathcal A(\square_2, B(\mathcal H)^+)$ be extremal and let $\rho=\bar w$ be the barycenter. Let $P=\mathrm{supp}(\rho)$ be the support projection of $\rho$. Let  $0\le E_{i,j}\le P$ be effects such that 
$\frac12w_{i,j}=\rho^{1/2}E_{i,j}\rho^{1/2}$. We will show that all $E_{i,j}$ are projections. Indeed,
let $M$ be an effect majorized by both $E_{0,0}$ and $I-E_{0,0}$. Then $\sigma:=\rho^{1/2}M\rho^{1/2}\le w_{0,0},w_{1,1}$, so that $\sigma\in L_{0,0}\cap L_{1,1}=\{0\}$. Since $M\le P$, it follows that $M=0$, so that $E_{0,0}\wedge (I-E_{0,0})=0$ and this implies that $E_{0,0}$ is a projection. Then $E_{1,1}=P-E_{0,0}$ is a projection as well, orthogonal to $E_{0,0}$. Similarly for $E_{0,1}$ and 
$E_{1,0}$.

Let $A\in B_h(P\mathcal H)$  be such that $A$ commutes with both $E_{0,0}$ and $E_{0,1}$. Then $\rho^{1/2}AE_{i,j}\rho^{1/2}\in L_{i,j}$ and $\rho^{1/2}A\rho^{1/2}\in (L_{0,0}+L_{1,1})\cap(L_{0,1}+L_{1,0})$. By \eqref{eq:extr_2}, this implies 
$\rho^{1/2}A\rho^{1/2}=t\rho$ for some $t\in \mathbb R$, so that $A=tP$.  Hence any element commuting with both $E_{0,0}$ and $E_{1,1}$ must be a multiple of $P$. But the commutant of two projections is always nontrivial, unless one of the following two cases occurs:
\begin{enumerate}
\item[(a)] $P$ is rank one and all $E_{i,j}$ are either $P$ or 0. Then $\rho$ is rank one and $W=\mea^i_j(\cdot)\rho$ for some $i,j\in \{0,1\}$, these are precisely the  ETB  extremal maps.
\item[(b)] $P$ is rank 2 and $E_{0,0}$ and $E_{0,1}$ are rank one  non-commuting projections. In that case, $\rho$ is rank 2 and 
all $w_{i,j}$ are rank 1 operators.
\end{enumerate}
 It follows that a non-ETB  map $W \in\mathcal A(\square_2, B(\mathcal H)^+) $ is extremal if and only if all its vertices are extremal.

\end{ex}

\subsection{Incompatibility degree}

Incompatibility degree of a collection of measurements  can be defined as the least amount of noise that has to be added to obtain a compatible collection.  Following \cite{hmz2016aninvitation}, the noise will have the form of coin-toss measurements, 
see \cite{bhss2013comparing, caskr2016quantitative,wopefe2009} for related definitions.  

For $p\in \Delta_{l}$, a coin-toss measurement is defined as a constant map $f_p(x)\equiv  p$. It is immediate that $f_p$ 
  is compatible with any $g\in \aff(K,\Delta_l)$ (since the map $(g,f_p)$ factorizes through $\Delta_l\times\Delta_0\simeq \Delta_l$). Let us again fix $l_0,\dots,l_k\in \mathbb N$ and let $p^i\in \Delta_{l_i}$, $0=1,\dots,k$. Then  the channel given by the collection of coin-tosses
 $(f_{p^0},\dots,f_{p^k})$ is the constant map $F_s(x)\equiv s:=(p^1,\dots,p^k)\in \Simp$.

For  $F\in \mathcal A(K, \Simp)$ and $s\in \Simp$,  we  define the incompatibility degree as
\[
ID_{s}(F):=\min\{\lambda\in [0,1],\ (1-\lambda)F+\lambda F_{s}\ \mbox{is ETB}\}.
\]  
We  also define
\[
ID(F):=\inf_{s\in \Simp} ID_{s}(F).  
\]
We first show that $ID(F)$ is attained at an interior point of $\Simp$.

\begin{lemma}\label{lemma:IDinf} $ID(F)=\inf_{s\in ri(\Simp)} ID_{s}(F)$.
\end{lemma}

\begin{proof} Let $s_0\in \partial \Simp$ and let $s_1\in ri(\Simp)$, then $s_t:= ts_1+(1-t)s_0\in ri(\Simp)$ for all $t\in (0,1]$. Put 
$\mu:=\frac{\lambda(1-t)}{1-\lambda t}$. We have $F_{s_t}=tF_{s_1}+(1-t)F_{s_0}$ and
\[
(1-\lambda) F+\lambda F_{s_t}= (1-\lambda t)\left((1-\mu)F+\mu F_{s_0}\right)+
\lambda tF_{s_1}.
\]
Assume that $ID_{s_0}(F)= \mu$, then $(1-\mu) F+\mu F_{s_0}$ is ETB, so that  $(1-\lambda) F+\lambda F_{s_t}$ is ETB as well. It follows that 
\[
\inf_{s\in ri(\Simp)} ID_s(F)\le ID_{s_t}(F)\le \lambda=\frac{ID_{s_0}(F)}{1-t(1-ID_{s_0}(F))}.
\]
Letting $t\to 0$ implies the result.

\end{proof}

 The next result shows that the incompatibility degree can be obtained  using incompatibility witnesses. Note that the minimum $q_s(F)$ below is attained at an extremal element in $\aff(\Simp, V(K)^+)$.

For pairs of two-outcome measurements, the following  expression for incompatibility degree 
 is related to the dual linear program of \cite{plavala2016allmeasurements, jepl2017conditions}.
In the quantum case, similar results using SDP were obtained in \cite{wopefe2009}.

\begin{prop}\label{prop:max_idegree} Let $s\in ri(\Simp)$.  Let us denote
\[
\mathcal W_{s}:=\{ W\in \mathcal A(\Simp, V(K)^+), \ W(s)\in K\}
\]
and for $F\in \mathcal A(K,\Simp)$, 
\[
q_{s}(F):= \min_{W\in \mathcal W_{s}} \Tr FW.
\]
Then  
\[
ID_{s}(F)=\left\{ \begin{array}{cc} 0 & \mbox{if } q_{s}(F)>0\\
                                    \frac{-q_{s}(F)}{1-q_{s}(F)} & \mbox{otherwise}.
\end{array}\right.
\]
\end{prop}

For the proof, we will need the following lemma.

\begin{lemma}\label{lemma:ou} Let $(V,V^+)$ be an ordered vector space and $u\in V$ an order unit. 
Let $\states(V,V^+,u)=\{\sigma\in (V^+)^*,\ \<\sigma,u\>=1\}$. 
Then for $v\in V$, 
\[
\inf\{t, v+tu\in V^+\}=\max_{\sigma\in \states(V,V^+,u)}-\<\sigma,v\>.
\]
\end{lemma}

\begin{proof} Let $t_0$ denote the infimum on the LHS and $s_0$ the supremum on the RHS.
Let $t\in \mathbb R$ be such that $v+tu\in V^+$. Then for any $\sigma\in \states(V,V^+,u)$, we have 
\[
0\le \<\sigma,v+tu\>=\<\sigma,v\> +t
\] 
so that $-\<\sigma,w\>\le t$. It follows that $s_0\le t_0$. Conversely, note that 
\[
\<\sigma, v+s_0u\>=\<\sigma,v\>+s_0\ge 0
\] 
for all $\sigma\in \states(V,V^+,u)$, hence for all elements in $(V^+)^*$, this implies $t_0\le s_0$. 

\end{proof}

\begin{proof}[Proof of Proposition \ref{prop:max_idegree}]  Note that $F_{s}= 1_K(\cdot)s$ is an interior element in the cone 
$\mathcal A_{sep}(K,V(\Simp)^+)$, hence   an order unit and we have  $\Tr F_sW=\<1_K,W(s)\>$.  
Clearly, if $q_s(F)>0$ then $F$ is compatible. Otherwise, by the above Lemma, $-q_s(F)$ is the smallest $t\ge 0$ such that
$F+tF_s$ is ETB, so that $ID_s(F)=\frac{-q_s(F)}{1-q_s(F)}$.

\end{proof}

We next use  the pair of dual bases in \eqref{eq:basis} and \eqref{eq:dualbasis} to find a suitable expression for $\Tr FW$.
 To shorten the notations, put $w^i_j:=w_{l_0,\dots,l_{i-1},j,l_{i+1},\dots,l_k}$, 
 note that $w^i_{l_i}=w_{l_0,\dots,l_k}$ for all $i$. Then  
 \[
  W(\edg^i_j)= w^i_j-w^i_{l_i},\qquad j=0,\dots,l_i, \ i=0\dots,k.
 \]
 We have
\begin{align}
\Tr FW&= \<1_\Simp, FW(\simp_{l_0,\dots,l_k})\>+\sum_{i=0}^k\sum_{j=0}^{l_i-1}\<\mea^i_j, FW(\edg^i_j)\>\\
&=\<1_K,w_{l_1,\dots,l_k}\>+\sum_{i=0}^k\sum_{j=0}^{l_i-1}\<f^i_j,w^i_j-w^i_{l_i}\>\notag \\ 
&= \sum_{i=0}^k \sum_{j=0}^{l_i}\<f^i_j,w^i_j\>-k\<1_K,w_{l_1,\dots,l_k}\>\label{eq:tracefw}
\end{align}

Using this, we  obtain another characterization of witnesses for two-outcome measurements. Recall that $\edg^0_0$,\dots,$\edg^k_0$ are edges adjacent to the vertex $\simp_{1,\dots,1}$. 

\begin{coro}\label{coro:2witnesses} Let  $W\in \mathcal A(\square_{k+1}, V(K)^+)$. Then $W$ is a witness if and only 
if
\[
\sum_{i=0}^k\|W(\edg^i_0)\|_K> 2\<1_K,\bar w\>.
\]
\end{coro}

\begin{proof}  We may assume that $\<1_K,\bar w\>=1$, so that $W\in \mathcal W_{\bar\simp}$. Put $\mu^i_j:=\<1_K,w^i_j\>$. Then
 for any $F\in \aff(K,\square_{k+1})$, we have using \eqref{eq:tracefw}
\begin{align*}
\Tr FW&=\Tr (F-F_{\bar \simp})W+\Tr F_{\bar\simp}W\\
&=\sum_{i=0}^k \sum_{j=0}^1(\<f^i_j,w^i_j\>- \frac12\mu^i_j)+1.
\end{align*}

$W$ is a witness if and only if
\begin{align*}
0&>\min_{F\in \mathcal A(K,\square_{k+1})} \Tr FW\\
&=\sum_{i=0}^k \min_{f\in E(K)}(\<f,w^i_0\>+\<1-f,w^i_1\>- \frac12(\mu^i_0+\mu^i_1))+1\\
&=\sum_{i=0}^k(\frac12(\mu^i_1-\mu^i_0)-\max_{f\in E(K)}\<f,w^i_1-w^i_0\>)+1\\&=-\frac12\sum_{i=0}^k \|W(\edg^i_0)\|_K+1
\end{align*}
The last equality follows from  the fact that $\max_{f\in E(K)} \<f,\psi\>=\frac12(\|\psi\|_K+\<1_K,\psi\>)$ for all $\psi\in V(K)$.

\end{proof}

Observe that by the above proof, the maximal value of $ID_{\bar \simp}$ attainable by $k+1$ two-outcome measurements on a state space $K$ is obtained from the maximal value of $\sum_{i=0}^k \|W(\edg^i_0)\|_K$ over all maps $W\in \aff(\square_{k+1},V(K)^+)$ with $\bar w\in K$. For $k=1$, the  maximal value of $ID_{\bar\simp}$ for quantum state spaces was obtained in \cite{bggk2013degree}.
 We prove this result by our method  and show that  $ID$ attains the same value.

\begin{coro}\label{coro:quantummax} For a quantum state space 
$\states=\states(\mathcal H)$, we have  
\[
\max_{F\in \mathcal A(\states,\square_2)} ID(F)=1-\frac1{\sqrt{2}}.
\]
\end{coro}

\begin{proof} Let $F\in\mathcal A(\states,\square_2)$ be incompatible and let $W\in \mathcal W_{\bar\simp}$ be such that 
$\Tr F W=q_{\bar\simp}(F)$. We may assume that $W$ is extremal in $\mathcal A(\square_2,B(\mathcal H)^+)$, so by 
 Example \ref{ex:quantumwitness}, there are some unit vectors $x_{i,j}\in \mathcal H$ and a rank 2 density operator
 $\rho$ with support projection $P$ such that 
 \[
 |x_{0,0}\>\<x_{0,0}|+|x_{1,1}\>\<x_{1,1}|= |x_{0,1}\>\<x_{0,1}|+|x_{1,0}\>\<x_{1,0}|=P
\]
 and the vertices of $W$ satisfy $\frac12w_{i,j}=\rho^{1/2}|x_{i,j}\>\<x_{i,j}|\rho^{1/2}$. By  H\"older's inequality, we have  
 \begin{align*}
\frac12\|W(\edg^0_0)\|_1&=\|\rho^{1/2}(|x_{0,1}\>\<x_{0,1}|-|x_{1,1}\>\<x_{1,1}|)\rho^{1/2}\|_1\\
&\le \||x_{0,1}\>\<x_{0,1}|-|x_{1,1}\>\<x_{1,1}|\|\\
&=\sqrt{1-|\<x_{0,1}|x_{1,1}\>|^2}=:c.
 \end{align*}
Similarly,
\begin{align*}
\frac12\|W(\edg^1_0)\|_1\le \sqrt{1-|\<x_{1,0}|x_{1,1}\>|^2}=:d.
 \end{align*}
 Since $\<x_{0,1},x_{1,0}\>=0$, we have $c^2+d^2=1$ and hence $c+d\le \sqrt{2}$.  
 By the proof of Corollary \ref{coro:2witnesses}, it follows that 
\begin{align*}
q_{\bar\simp}(F)&=\Tr F W\ge -\frac12(\|W(\edg^0_0)\|_1+\|W(\edg^1_0)\|_1)+1\\
&\ge 1-(c+d)\ge 1-\sqrt{2}.
\end{align*}
On the other hand, let $W_0$ have vertices $\rho_{i,j}=|x_{i,j}\>\<x_{i,j}|$, with unit vectors $x_{i,j}$ such that $|\<x_{0,1}|x_{1,1}\>|=|\<x_{1,0}|x_{1,1}\>|=\sqrt{2}/2$. In this case, $W(s)\in \states$ for any $s\in \square_2$. 
Let $F_0=(f^0,f^1)$ be determined by the effects
\begin{align*}
f^1_0&:=\mathrm{argmax}_{f\in E(\states)} \<f, \rho_{1,1}-\rho_{0,1}\>,\\ f^2_0&:=\mathrm{argmax}_{f\in E(\states)} \<f, \rho_{1,1}-\rho_{1,0}\>.
\end{align*}
Then we have
\begin{align*}
1-\sqrt{2}&\le q_{\bar\simp}(F_0)\le\Tr F_0 W_0\\
&=-\frac12(\|\rho_{0,1}-\rho_{1,1}\|_1+\|\rho_{1,0}-\rho_{1,1}\|_1)+1\\ &=1-\sqrt{2}
\end{align*}
It follows that $\max_{F} ID_{\bar\simp}(F)=1-\frac1{\sqrt{2}}$, 
this corresponds to the results obtained in \cite{bggk2013degree}. Note further that the witness $W_0\in \mathcal W_s$ for any $s\in ri(\square_2)$. We obtain
\[
q_s(F_0)=\min_{W\in \mathcal W_s}\Tr F_0 W\le \Tr F_0 W_0=q_{\bar\simp}(F_0).
\]
It follows that 
\[
ID_{\bar\simp}(F_0)=\frac{-q_{\bar\simp}(F_0)}{1-q_{\bar\simp}(F_0)}\le \frac{-q_{s}(F_0)}{1-q_{s}(F_0)}=ID_s(F_0)
\] 
so that $ID_{\bar\simp}(F_0)=ID(F_0)$. We then have
\[
ID(F)\le ID_{\bar\simp}(F)\le ID_{\bar\simp}(F_0)=ID(F_0)=1-\frac1{\sqrt{2}}.
\]

\end{proof}

\subsection{Maximally incompatible measurements}

Let $\Simp$ be any polysimplex and let $F=(f^0,\dots,f^k)\in \aff(K,\Simp)$ be a collection of measurements. It is well known that  
\[
ID_{s}(F)\le \frac{k}{k+1}
\]
for any $s\in \Simp$  (see e.g. \cite{hmz2016aninvitation}): the joint measurement for
$\tfrac{1}{k+1} F+\tfrac{k}{k+1} F_s$ can be defined  by choosing one of the measurements uniformly at random and replace all other measurements by the corresponding  coin tosses. If $ID(F)=\frac{k}{k+1}$, we say that $F$ is maximally incompatible. 
 We now give a general characterization of maximal incompatibility.

\begin{thm}\label{thm:maximalID} Let $F\in \mathcal A(K,\Simp)$. Then the following are equivalent.
\begin{enumerate}
\item[(i)] $F$ is maximally incompatible.
\item[(ii)] $ID_s(F)=\frac{k}{k+1}$ for all $s\in \Simp$.
\item[(iii)] $ID_s(F)=\frac{k}{k+1}$ for some $s\in ri(\Simp)$.
\item[(iv)] There is some  $W\in \aff(\Simp,K)$ such that $\Tr FW=-k$.
\item[(v)] There is  some $W\in \aff(\Simp,K)$ such that 
\[
\<f^i_j, w_{n_1,\dots,n_{i-1},j,n_{i+1},\dots,n_k}\>=0,\quad \forall i,j;\ n_0,\dots,n_k.
\]
\end{enumerate}

\end{thm}

\begin{proof} (i) $\implies$ (ii) follows from the definition of $ID(F)$ and the fact that $ID_s(F)\le \frac{k}{k+1}$ for all $s$,
 (ii) $\implies$ (iii) is trivial. Assume (iii), then by Proposition \ref{prop:max_idegree} there is some $W\in \mathcal W_s$ such that
 $\Tr FW=q_s(F)=-k$. For any $n_0,\dots,n_k$, choose the pair of dual bases of $V(\Simp)$ and $A(\Simp)$  by fixing the vertex $\simp_{n_0,\dots,n_k}$ as in Remark \ref{rem:basis}, then exactly as in  \eqref{eq:tracefw}, we obtain
 \begin{align*}
-k&= \sum_{i=0}^k\sum_{j=0}^{l_i} \<f^i_j,w_{n_0,\dots,n_{i-1},j,n_{i+1},\dots,n_k}\>-k\<1_K,w_{n_0,\dots,n_k}\>\\&\ge-k\<1_K,w_{n_0,\dots,n_k}\>
 \end{align*}
 This implies that $\<1_K,w_{n_0,\dots,n_k}\>\ge 1$ for all $n_0,\dots,n_k$. On the other hand, since $W(s)\in K$ is a convex combination  of all $w_{n_0,\dots,n_k}$ with nonzero coefficients, we must have $\<1_K,w_{n_0,\dots,n_k}\>=1$, hence (iv) holds.
 Further, if $W$ is as in (iv), the inequality in the above computation must be an equality, so that $\<f^i_j,w_{n_0,\dots,n_{i-1},j,n_{i+1},\dots,n_k}\>=0$ for all $i$ and $j$,
  hence (v) holds.

  Assume (v), then $W(s)\in K$ for any $s\in \Simp$ and $\Tr FW=-k$ by \eqref{eq:tracefw}. It follows that $ID_s(F)$ is maximal for all $s\in ri(\Simp)$. By Lemma \ref{lemma:IDinf}, this implies (i).

\end{proof}

Maximal incompatibility has a nice geometric interpretation for two-outcome  measurements. For $k=1$ the following results were essentially proved in \cite{jepl2017conditions}. Recall the definition of retraction-section pairs in Section \ref{sec:poly}. 

\begin{coro}\label{coro:maxcube} Let $F\in \aff(K,\square_{k+1})$. 
Then $F$ is maximally incompatible  if and only if $F$ is a retraction. 

\end{coro}

 \begin{proof} Assume that $F$ is a retraction and let $S:\Simp\to K$ be the corresponding section. Let $U$ be the automorphism 
 of $\Simp$ given as $U(\simp_{n_0,\dots,n_k})=\simp_{1-n_0,\dots,1-n_k}$ and put $W=SU$. Then $W\in\aff(\Simp, K)$ and we have
 \[
\Tr FW=\Tr FSU=\Tr U=1+\sum_{i=0}^k\<\mea^i_0,U(\edg^i_0)\>=-k.
 \]
Hence $F$ is maximally incompatible. Conversely, assume that $F$ is maximally incompatible and let $W$ be the witness as in Thm. 
\ref{thm:maximalID} (v). 
Observe that then 
\[
\mea^i_jFW=f^i_jW=\mea^i_{1-j},
\]
it follows that $FW =U$. Putting $S=WU$ we obtain $FS=id_\Simp$, so that $F$ is a retraction.

\end{proof}

\begin{coro}\label{coro:maxcube_exist} There exist $k+1$ maximally incompatible  two-outcome measurements on $K$ if and only if there exists a projection $K\to K$ whose range is affinely isomorphic to $\square_{k+1}$.

\end{coro}

\begin{proof}  If $F$ is maximally incompatible, then by Corollary \ref{coro:maxcube} there is a section $S\in \aff(\square_{k+1},K)$ such that $FS=id$. It follows that $P:=SF$ is a projection $K\to K$
 such that $PS=S$ and $FP=F$, hence the restriction of $F$ to the range of $P$ is an isomorphism onto $\square_{k+1}$ whose inverse is $S$.

Conversely, assume that $P:K\to K$ is such a projection and let $U:P(K)\to \square_{k+1}$ be the isomorphism onto the cube. Then 
$F=UP$ is obviously a retraction, since then $FU^{-1}=UPU^{-1}=UU^{-1}=id_{\square_{k+1}}$.

\end{proof}

\begin{ex}\label{ex:maximalS}(\textbf{Maximal incompatibility in $\Simp$}) We will show that any collection of effects 
$\mea^0_{n_0},\dots,\mea^k_{n_k}$ with $n_i\in \{0,\dots,l_i\}$ is  maximally incompatible. Indeed, let  $F\in \aff(\Simp, \square_{k+1})$ be the corresponding  channel and let $S\in \aff(\square_{k+1},\Simp)$ be determined by the collection of measurements $t^0,\dots,t^k$, given by
\[
t^i(s)=\mea^i_0(s)\delta^i_{n_i}+(1-\mea^i_0(s))\delta^i_{n_i'},\qquad s\in \square_{k+1},
\]
for some $n_i'\ne n_i$. Then $FS=id_{\square_{k+1}}$, so that $F$ is a retraction. Note that this also implies that the projections 
 $\mea^0,\dots,\mea^k$ of $\Simp$ are maximally incompatible as well, since they determine the identity map $id_\Simp$ and 
 if $W$ is a witness as in Theorem \ref{thm:maximalID} (iv),
  then $\Tr id_\Simp FW=\Tr FW=-k$ and $FW\in \aff(\Simp,\Simp)$.

  \end{ex}

We next show that  maximally incompatible  measurements exist in the space of quantum 
channels, cf. \cite{srcz2016incompatible, jepl2017conditions}. This result  is a simple consequence of Proposition \ref{prop:qchannel_polysimp} and Example \ref{ex:maximalS}, see also Remark \ref{rem:channels}.

\begin{coro}\label{coro:maximalQC}  There exists a maximally incompatible collection of $d_A$ two-outcome measurements on $\Ce_{A,A'}$.
\end{coro}

\begin{proof} Let $R\in \aff(\Ce_{A,A'},\Delta^{d_A}_{d_{A'}-1})$ be the retraction as in  Proposition \ref{prop:qchannel_polysimp} and let $F\in \aff(\Delta^{d_A}_{d_{A'}-1},\square_{d_A})$ be as in Example \ref{ex:maximalS}. Then  $F$ is a retraction, so that $FR\in \aff(\Ce_{A,A'}, \square_{d_A})$ is a retraction as well.  

\end{proof}

\begin{rem} As in the above proof, any retraction $K\to \Simp$ is maximally incompatible. On the other hand, let $F'=(f^0,\dots,f^k)$ 
 be such that the collection $G$ of two-outcome measurements determined by the effects $f^0_{n_0},\dots, f^k_{n_k}$ is maximally incompatible. Then $F'$ is maximally incompatible as well. Indeed, with the notation of Example \ref{ex:maximalS}, $G=FF'$. Let $W'$ be the witness $\square_{k+1}\to K$ such that $\Tr GW'=-k$, then $W'F:\Simp\to K$ is a witness demonstrating maximal incompatibility of $F'$.
Since the other effects are not involved, $F'$ is not necessarily a retraction, so that Corollary \ref{coro:maxcube} cannot be extended to all polysimplices.

\end{rem}

\section{Steering and nonlocality}

Quantum steering refers to the property of entangled quantum states which allows one to ''steer'' the state 
of one component by choosing suitable measurements on the other \cite{schrodinger1936probability}. 
A rigorous operational definition was given in \cite{wjd2007steering} and can be easily rephrased in the setting of GPT.

\subsection{Steering in GPT}

Let $K_A,K_B$ be state spaces and let  $y\in K_A\widetilde \otimes K_B$ be a joint state. If a measurement  $f_A\in \aff(K_A,\Delta_n)$ is applied on system $A$, then $y$ is mapped onto some element $(f_A\otimes id_B)(y)\in \Delta_n\otimes K_B$. It means that there are some states $x_j\in K_B$ and a probability measure $p\in \Delta_n$ such that $(f_A\otimes id_B)(y)=\sum_{j=0}^n p(j)\delta_j\otimes x_j$.
This has the interpretation that with probability $p(j)$, the outcome $j$ is observed on $A$  and the state of $B$ turns into $x_j$. The collection  $\{p(j),x_j\}$ of states and probabilities is called an ensemble. If $B$ has no information about the outcome, the state of $B$ is just the average state $\sum_jp(j)x_j=(1_A\otimes id_B)(y)=:y_B$.  

Assume now that an observer on the system $A$ can choose from a collection of measurements $f^i_A\in \aff(K_A,\Delta_{l_i})$. Then 
we obtain a set of ensembles  $\{p(j|i), x_{j|i}\}$ with a common average state $y_B$. Such a set  is called an assemblage. According to \cite{wjd2007steering}, an assemblage does not demonstrate steering if there is a (finite) set $\Lambda$ 
of ''classical messages''  distributed according to a probability measure $q$, corresponding  elements  
$\{x_\lambda\in K_B, \ \lambda \in \Lambda\}$ and conditional probabilities $q(j|i,\lambda)$ such that 
\begin{equation}\label{eq:LHS}
p(j|i)x_{j|i}=\sum_\lambda q(\lambda) q(j|i,\lambda)x_\lambda.
\end{equation}
In this case, the assemblage can be explained by a local hidden state model, see \cite{wjd2007steering} for more details. 
The next result shows that steering can be conveniently expressed in terms of the minimal and maximal tensor products of 
compact convex sets.

\begin{thm}\label{thm:steering} Let $K$ be a state space and $\Simp$ a polysimplex. Let $\beta\in \Simp\widehat\otimes K$.
\begin{enumerate}
\item[(i)] There is an assemblage $\{p(j|i), x_{j|i}, j=0,\dots,l_i, i=0,\dots,k\}$ of elements in  $K$ with average state $\sum_j p(j|i)x_{j|i}=x\in K$, such that 
\begin{equation}\label{eq:assemblage}
\beta=\simp_{l_0,\dots,l_k}\otimes x+ \sum_{i=0}^k\sum_{j=0}^{l_i-1} \edg^i_j\otimes p(j|i)x_{j|i}.
\end{equation}

\item[(ii)] The assemblage in (i) does not demonstrate steering if and only if $\beta$ is separable.

\end{enumerate}
Moreover, any element of the form \eqref{eq:assemblage} is in  $\Simp\widehat\otimes K$.
\end{thm}

\begin{proof} Using the basis \eqref{eq:dualbasis}, we have 
\[
\beta=\simp_{l_0,\dots,l_k}\otimes \phi+\sum_{i=0}^k\sum_{j=0}^{l_i-1}\edg^i_j\otimes \phi^i_j
\]
for some $\phi,\phi^i_j\in V(K)$. By definition of $\Simp\widehat\otimes K$, we must have $\<\beta, \mea^i_j\otimes f\>\ge 0$ for 
 all $i,j$ and $f\in E(K)$. For $j\ne l_i$ this is true if and only if  $\phi^i_j\in V(K)^+$.  We also have for all $i$ and $f\in E(K)$
\[
\sum_{j=0}^{l_i-1}\<\phi^i_j,f\>\le \sum_{j=0}^{l_i}\<\beta,\mea^i_j\otimes f\>=\<\beta,1_\Simp\otimes f\>=\phi,
\]
hence $\sum_j\phi^i_j\le \phi\in V(K)^+$ and  $\<\beta,1_\Simp\otimes 1_{K}\>=\<\phi,1_{K}\>=1$. Put $\phi^i_{l_i}:=\phi-\sum_{j=0}^{l_i-1}\phi^i_j$ and $p(j|i):=\<1_K,\phi^i_j\>$, $x_{j|i}:= p(j|i)^{-1}\phi^i_j$ (if $p(j|i)>0$, otherwise $x(j|i)$ can be anything) 
for all $i$ and $j$. Then $\{p(j|i),x_{j|i}\}$ is an assemblage with average state $x:=\phi$. This proves (i).

For (ii), assume that \eqref{eq:LHS} holds. Then $x=\sum_\lambda q(\lambda)x_\lambda$ and we have
\[
\beta=\sum_\lambda q(\lambda)s_\lambda\otimes x_\lambda,
\]
where $s_\lambda:=\simp_{l_0,\dots,l_k}+\sum_{i=0}^k\sum_{j=0}^{l_i-1} q(j|i,\lambda)\edg^i_j\in \Simp$, this follows from
 $\<\mea^i_j,s_\lambda\>=q(j|i,\lambda)\ge 0$ and $\<1_\Simp,s_\lambda\>=\<1_\Simp,\simp_{l_0,\dots,l_k}\>=1$. Hence $\beta$ is separable.  
Conversely, let 
\[
\beta=\sum_{n_0,\dots, n_k}\simp_{n_0,\dots,n_k}\otimes\alpha_{n_0,\dots,  n_k}
\]
for some $\alpha_{n_0,\dots,n_k}\in V(K)^+$.  Put  $\Lambda:=\{(n_0,\dots, n_k), n_i=0,\dots, l_i\}$,
$q(n_0,\dots,n_k):=\<1_{K},\alpha_{n_0\dots n_k}\>$, $x_{n_0,\dots,n_k}:= q(n_0,\dots,n_k)^{-1}\alpha_{n_0\dots n_k}$, 
and 
\[
q(j|i,(n_0,\dots,n_k)):=\left\{\begin{array}{cc} 1 & \mbox{ if } n_i=j\\ 
    0 & \mbox{otherwise}.\end{array}\right.
\]
 Then for all $i$ and $j$, 
 \begin{align*}
p(j|i)x_{j|i}&=\<\beta,\mea^i_j\otimes \cdot\>=\sum_{n_0,\dots,n_k, n_i=j} \alpha_{n_0,\dots,n_k}\\ &=
\sum_{\lambda\in \Lambda} q(\lambda)q(j|i,\lambda) x_\lambda.
 \end{align*}
This proves (ii). The last statement follows from the proof of (i).

 \end{proof}

In view of the preceding theorem, any element in $\Simp\widehat\otimes K$ will be called an assemblage. The common average state $x$ will be called the barycenter of $\beta$.

It is already known that 
there is no steering if $y$ is separable or the measurements $f^i_A$ are compatible. This can be seen immediately from Theorem 
\ref{thm:steering}. Let $F_A=(f^0_A,\dots, f^k_A)$, then $\beta=(F_A\otimes id_B)(y)\in \Simp\widehat \otimes K_B$ is the associated assemblage. If $y\in K_A\otimes K_B$, then $\beta$ must be separable and hence does not demonstrate steering. If the measurements are compatible, then $\beta$ is separable by Theorem \ref{thm:incompatibility}. 

In general, not all  assemblages in $\Simp\widehat \otimes K_B$ are obtained from some collection $F_A$ and  a bipartite state $y\in K_A\widetilde\otimes K_B$. If this is the case, we say that the state space $K_B$ admits steering. Note that quantum state spaces satisfy this condition. 
Somewhat stronger conditions  were studied in \cite{bgw2013ensemble} and their relations to homogeneity and weak self-duality 
of the state spaces were found.

\subsection{Steering degree}

Let $\beta\in \Simp\widehat \otimes K$ be an assemblage with barycenter $x$. A steering degree can be defined similarly as incompatibility degree, as the smallest amount of noise that has to be added to $\beta$ to obtain a separable element. For the noise, we use assemblages of the form $s\otimes x$ for 
$s\in \Simp$. This is a separable assemblage with $p(j|i)=\mea^i_j(s)$ and  $x_{j|i}=x$ for all $i,j$. 
We put 
\[
SD_{s}(\beta):=\min\{\lambda\in [0,1],\ (1-\lambda)\beta+\lambda s\otimes x\in \Simp\otimes K\}
\]  
and 
\[
SD(\beta):=\inf_{s\in \Simp} SD_{s}(\beta).  
\]
Observe that for any $F_A\in \aff(K_A,\Simp)$ and $y\in K_A\widetilde \otimes K_B$, we have 
\begin{equation}\label{eq:sd_le_id}
SD_s((F_A\otimes id_B)(y))\le ID_s(F_A).
\end{equation}
To see this, note that the barycenter of $(F_A\otimes id_B)(y)$ is the marginal $y_B$ and $s\otimes y_B= (F_s\otimes id_B)(y)$. 
We therefore have
\begin{align*}
(1-\lambda)(F_A\otimes id_B)(y)&+\lambda s\otimes y_B\\ &=\left(((1-\lambda)F_A+\lambda F_s)\otimes id_B \right)(y)
\end{align*}
and this is separable if $(1-\lambda)F_A+\lambda F_s$ is ETB. The possibility of attaining equality depends on the properties of $K_A$ and the form of composite state spaces in the GPT. Assume that $V(K_A)^+$ is weakly self-dual, so that there is an affine isomorphism 
 $\Psi:A(K_A)^+\to V(K_A)^+$. With the notations of Appendix \ref{app:positive}, we have  $(id\otimes \Psi)(\chi_{K_A})\in K_A\widehat \otimes K_A$, see Lemma \ref{lemma:chiKK}. We are now prepared to state the following result.

 \begin{thm}\label{thm:sd_id} Assume that there is an isomorphism  $\Psi:A(K_A)^+\to V(K_A)^+$ such that 
 $(id\otimes \Psi)(\chi_{K_A})\in K_A\widetilde \otimes K_A$.
  Then for any polysimplex $\Simp$, $F_A\in \aff(K,\Simp)$ and $s\in \Simp$, we have
  \[
  \sup_{y\in K_A\widetilde\otimes K_A} SD_s((F_A\otimes id_A)(y))=ID_s(F_A).
\]
 \end{thm}
\begin{proof} By \eqref{eq:sd_le_id}, the supremum on the left is never larger than $ID_s(F_A)$. Put $y=(id\otimes \Psi)(\chi_{K_A})$ and let $F_\lambda:= (1-\lambda)F_A+\lambda F_s$. Then $(F_\lambda\otimes id)(y)$ is separable if and only if 
$(F_\lambda\otimes id)(\chi_{K_A})$ is separable, which by Proposition \ref{prop:ETBapp} (ii) means that $F_\lambda$ is ETB. 
It follows that $SD_s((F_A\otimes id_A)(y))\ge ID_s(F_A)$.

\end{proof}

The conditions in the previous theorem are fulfilled in quantum state spaces, where $y$ is a pure maximally entangled state. 
In this case, this result was proved in \cite[pp. 8-9]{caskr2016quantitative}.

\begin{rem}\label{rem:steer_witness} Similarly as for incompatibility, we may define steering witnesses and their relation to steering degree, maximal steering  degree, etc. The witnesses will now be elements in $A(\Simp\otimes K)^+$. We will not investigate this here, only remark that since 
 the assemblages generate all of the positive cone $V(\Simp\widehat\otimes K)^+$, any non-separable element in $A(\Simp\otimes K)^+$ is a witness.

\end{rem}

\subsection{Nonlocality and Bell's inequalities}

Let $f^i_A\in \aff(K_A,\Delta_{l^A_i})$, $i=0,\dots,k_A$ and $f^i_B\in \aff(K_B,\Delta_{l^B_i})$, $i=0,\dots,k_B$, and let $y\in K_A\widetilde\otimes K_B$. If a measurement $f^{i_A}_A$ is chosen for $A$ and $f^{i_B}_B$ for $B$, then the result is a pair $(j_A,j_B)$ with probability
\[
p(j_A,j_B| i_A,i_B):=\<(f^{i_A}_A)_{j_A}\otimes (f^{i_B}_B)_{j_B},y\>.
\]
These conditional probabilities satisfy the no-signalling properties
\begin{align}
\sum_{j_A} p(j_A,j_B| i_A,i_B)&=p_B(j_B|i_B),\quad \forall i_A\label{eq:no_signallingA}\\  
\sum_{j_B} p(j_A,j_B| i_A,i_B)&=p_A(j_A|i_A), \quad \forall i_B\label{eq:no_signallingB}
\end{align}
where $p_A(j_A|i_A):=\<(f^{i_A}_A)_{j_A},y_A\>$, $p_B(j_B|i_B):=\<(f^{i_B}_B)_{j_B},y_B\>$. 
Following \cite{wjd2007steering}, we say that the state $y$ is Bell local if for all  measurements $f^i_A$ and $f^i_B$,
these probabilities admit a local hidden variable  (LHV) model, that is, there is a probability distribution $q$ on a set $\Lambda$ and 
 conditional probabilities $q_A(j_A|i_A,\lambda)$ and $q_A(j_B|i_B,\lambda)$ such that 
\begin{equation}\label{eq:LHV}
p(j_A,j_B|i_A,i_B)=\sum_\lambda q(\lambda) q_A(j_A|i_A,\lambda)q_B(j_B|i_B,\lambda)
\end{equation}
Let $F_A=(f_A^0,\dots,f_A^{k_A})$ and let $\Simp_A$  be the related polysimplex, similarly define $F_B$  and $\Simp_B$. Then 
$\gamma:=(F_A\otimes F_B)(y)\in \Simp_A\widehat\otimes \Simp_B$ and 
\[
(\mea^{i_A}_{j_A}\otimes \mea^{i_B}_{j_B})(\gamma)=\<(f^{i_A}_A)_{j_A}\otimes (f^{i_B}_B)_{j_B},y\>=p(j_A,j_B| i_A,i_B).
\]
It can be seen by putting $K_B=\Simp_B$ in Theorem \ref{thm:steering} that the elements $\gamma \in \Simp_A\widehat\otimes \Simp_B$ 
are characterized by the property that $p(j_A,j_B|i_A,i_B):=(\mea^{i_A}_{j_A}\otimes \mea^{i_B}_{j_B})(\gamma)$ are no-signalling conditional probabilities  and  \eqref{eq:LHV} describes precisely the separable elements. 
The tensor product $\Simp_A\widehat\otimes \Simp_B$ is therefore called the no-signalling polytope and $\Simp_A\otimes \Simp_B$ the local polytope.  

The steering witnesses in this case (see Remark \ref{rem:steer_witness}) will be called Bell witnesses. These are precisely the elements of 
$A(\Simp_A\otimes \Simp_B)^+$ that are not separable. With some normalization, there is a finite number of extremal Bell witnesses $\mu_1,\dots,\mu_N$ that completely determine the local polytope: if $\gamma\in \Simp_A\widehat \otimes \Simp_B$, then $\gamma$ is local if and only if
\begin{equation}\label{eq:bell}
\<\mu_i,\gamma\>\ge 0,\qquad  i=1,\dots,N.
\end{equation}
These are the Bell inequalities. Lemma \ref{lemma:chiKK} (iv) shows that, similarly as in the case of incompatibility, the Bell witnesses correspond to affine maps of the polysimplex $\Simp_A$ into a positive cone, this time it is the cone $A(\Simp_B)^+$. 
All Bell inequalities are given by extremal non-ETB elements in $\aff(\Simp_A,A(\Simp_B)^+)$.

\begin{ex} (\textbf{The CHSH inequality})
\label{ex:chsh} Assume that there is a pair of two-outcome measurements on both sides, so that 
 $\Simp_A=\Simp_B=\square_2$. Since $V(\square_2)^+\simeq A(\square_2)^+$, we see by Example \ref{ex:squarewitness} that 
  all extremal witnesses   are precisely (multiples of) the isomorphisms $\Psi_{i,j,k}\in \aff(\square_2, A(\square_2)^+)$ that map the extreme points of $\square_2$ to the four effects $\mea^i_j$:
\[
 \simp_{0,0}\mapsto \mea^i_{j},\ \simp_{1,1}\mapsto \mea^i_{1-j},\ \simp_{0,1}\mapsto \mea^{1-i}_{k},\ \simp_{1,0}\mapsto \mea^{1-i}_{1-k}.
\]
Let $\mu_{i,j,k}$ be the witness corresponding to $\Psi_{i,j,k}$. Using the basis elements \eqref{eq:basis} and \eqref{eq:dualbasis}, 
we get
\begin{align*}
\mu_{i,j,k}&= \mea^i_{1-j}\otimes 1_{\square_2}+ (\mea^{1-i}_k-\mea^i_{1-j})\otimes \mea^0_0\\ &+ (\mea^{1-i}_{1-k}-\mea^i_{1-j})\otimes \mea^1_0.
\end{align*}
 Let $F_A=(f^0_A,f^1_A)\in \aff(K_A,\square_2)$,  $F_B=(f^0_B,f^1_B)\in \aff(K_B,\square_2)$ and put 
 \begin{align*}
 a_1&:=1-2(f^1_A)_0, \ a_2:=1-2(f^0_A)_0,\\
 b_1&:=1-2(f^0_B)_0,\ b2:=1-2(f^1_B)_0,
 \end{align*}
then we can see that 
\begin{align}
\<\mu_{0,1,0}, (F_A\otimes F_B)(y)\>&=\<(F_A\otimes F_B)^*(\mu_{0,1,0}),y\>\notag\\
&=\frac12\left(1-\frac12\mathbb B\right)\label{eq:chsh}
\end{align}
where $\mathbb B=\<a_1\otimes(b_1+b_2)+a_2\otimes(b_1-b_2),y\>$, so that  \eqref{eq:bell} becomes the CHSH inequality.

\end{ex}

\subsection{Bell inequalities and the incompatibility degree}

The maximal value of $\mathbb B$ in Example \ref{ex:chsh} that can be attained by two-outcome measurements is called the CHSH bound.
It is well known that the outcome probabilities satisfy the LHV model \eqref{eq:LHV} if and only if $\mathbb B\le 2$ and we always have 
 $\mathbb B\le 4$. In quantum state spaces, the Tsirelson bound holds:  $\mathbb B\le 2\sqrt{2}$. It was observed in  \cite{wopefe2009} that the incompatibility degree for pairs of quantum effects is connected to this bound. 

The relation of incompatibility degree and CHSH bound in GPT  was proved in \cite{buste2014steering}. 
We include a proof in our setting.

\begin{thm}\label{thm:bell_id} Let $\Simp_A, \Simp_B$ be polysimplices. Let $F_A\in \aff(K_A,\Simp_A)$, 
$F_B\in \aff(K_B,\Simp_B)$, $y\in K_A\widetilde\otimes K_B$ and assume that $F_A$ is incompatible. Then for any $\mu\in A(\Simp_A\otimes \Simp_B)^+$ and $s\in ri(\Simp_A)$,  we have 
\[
\<\mu,F_A\otimes F_B(y)\>\ge \|\mu\|_{max} q_{s}(F_A). 
\]
If $K_A$ admits steering and $\Simp_A=\square_2$, then there is some state space $K_B$, 
$F_B\in (K_B,\square_2)$ and $y\in K_A\widetilde\otimes K_B$ such that 
\[
\<\mu_{i,j,k},F_A\otimes F_B(y)\>=\frac12 q_{\bar \simp}(F_A) 
\]
for the witness $\mu_{i,j,k}$  as in example \ref{ex:chsh}.
\end{thm}

\begin{proof}
By Lemma \ref{lemma:chiKK} (iv), there are  some $T\in \aff(K_B^\dagger, V(K_A)^+)$ and $M\in \aff(\Simp_A, V(\Simp_B^\dagger)^+)$ 
 such that  $y=(T\otimes id)(\chi_{K_B^\dagger}) $ and $\mu=(M^*\otimes id)(\chi_{\Simp_B})$.
Then by using Lemma \ref{lemma:chiKK} (iii), 
\begin{align*}
\<\mu,F_A\otimes F_B(y)\>&=\<(M^*\otimes id)(\chi_{\Simp_B}), (F_AT\otimes F_B)(\chi_{K_B^\dagger})\>\\
&=\<\chi_{\Simp_B}, (MF_ATF_B^*\otimes id)(\chi_{\Simp_B}^\dagger)\>\\
&=\Tr F_ATF_B^*M.
\end{align*}
Put $W:=TF_B^*M$, then $W\in \aff(\Simp_A,V(K_A^+))$. Moreover, 
\begin{align*}
\Tr F_{s} W&=\<1_{K_A}, W(s)\>=\<F_BT^*(1_{K_A}),M(s)\>\\
&=\<s',M(s)\>,
\end{align*}
where $s':=F_BT^*(1_{K_A})$. It is easy to see that $T^*(1_{K_A})\in K_B$, so that $s'\in \Simp_B$.
 It follows that 
\begin{align*}
t&:=\Tr F_{s} W=\<s',M(s)\>=\<M^*(s'),s\>=\mu(s\otimes s')\\
&\le \|\mu\|_{max}.
\end{align*}
We have  $t^{-1}W\in \mathcal W_s$, so that 
\[
t^{-1}\<\mu,F_A\otimes F_B(y)\>=t^{-1}\Tr F_AW\ge q_s(F_A).
\]
The final inequality follows by the assumption that $F_A$ is incompatible, so that $q_s(F)<0$.

Assume now that $K_A$ admits steering. Let $W\in \aff(\square_2, V(K_A)^+)$  be a witness in $\mathcal W_{\bar\simp}$ such that 
$\Tr F_AW=q_{\bar\simp}(F_A)$. In view of the above proof, it is enough to show that $W=TF_B^*M$ for suitable $T$, $F_B$ and $M$.
So let $M=\Psi_{i,j,k}$ be the isomorphism as in Example \ref{ex:chsh}. 
Put $\beta:=(WM^{-1}\otimes id)(\chi_{\square_2^\dagger})$, then clearly $\beta\in V(K_A\widehat\otimes\square_2 )^+$ and 
\[
\<\beta,1\otimes 1\>= \<1_{K_A},WM^{-1}(1_{\square_2})\>= 2\<1_{K_A},W(\bar\simp)\>=2.   
\]
Hence $\frac12 \beta$ is an assemblage. Since $K_A$ admits steering, there is some state space $K_B$, $F_B\in \aff(K_B,\square_2)$ and   
$y\in K_A\widetilde \otimes K_B$ such that 
\[
\frac12 \beta=(id\otimes F_B)(y)=(T\otimes F_B)(\chi_{K_B^\dagger})=(TF_B^*\otimes id)(\chi_{\square_2^\dagger})
\]
It follows that $\frac12 W=TF_B^*M$, this finishes the proof.

\end{proof}

Note that the crucial part of the proof of the equality in the above theorem is that $V(\square_2)^+$ is weakly self-dual, so we may chose $M$ to be an isomorphism. This is not true for any other $\Simp_A$ and $\Simp_B$. As we have seen in the proof, the Bell scenario provides incompatibility witnesses only of the form $W=TF_B^*M$, that is, factorizing through some $A(\Simp_B)^+$, which can be 
  weaker for detection of some types of incompatibility. This seems to be the reason why already for three quantum effects, incompatibility in some cases cannot be detected by violation of Bell inequalities. This was observed in \cite{quvebr2014joint} in the case of qubit states, but  the above arguments suggest that such effect exist in any non-classical theory in our class of GPT. 

\subsection{Nonlocality in spaces of quantum channels \label{sec:nonlocQC}}

It was proved in  \cite{bgnp2001causal} that one can obtain probabilities maximally violating the CHSH inequality, that is, attaining the value $\mathbb B=4$, by using causal 
 bipartite quantum channels. In fact, it was shown recently in \cite{hosa2017achannel} that one can obtain all no-signalling probabilities in this way. In these works, the GPT setting was not used, but nevertheless it was shown  that any element of the no-signalling polytope can be 
 obtained by applying sets of channel measurements to both parts of an element of $\Ce^{caus}_{AB,A'B'}$.
 Maximal violation of the CHSH inequality in spaces of quantum channels was also proved using GPT in \cite{plzi2017popescu}. Note that 
the channel used in \cite{bgnp2001causal, plzi2017popescu}  was a bipartite c-c channel. 

The aim of the present paragraph is to remark that this feature of quantum channels is immediate from Proposition \ref{prop:qchannel_polysimp}. Indeed, the isomorphism  $\Delta_n^{k+1}\simeq \aff(\Delta_{k},\Delta_n)$ also implies that 
\[
\Delta_{n_A}^{k_A+1}\widehat\otimes \Delta_{n_B}^{k_B+1}\simeq
\aff(\Delta_{k_A},\Delta_{n_A})\widehat\otimes\aff(\Delta_{k_B},\Delta_{n_B}).
\]
In this way, any no-signalling polytope is isomorphic to a face in 
the space of causal  classical bipartite channels and the local polytope corresponds to the local channels in this face
(see also Example \ref{ex:composite_QC}).

Let now $(R_A,S_A)$, $(R_B,S_B)$ be the retraction-section pairs as in Proposition \ref{prop:qchannel_polysimp}
 and let $\gamma\in  \Delta^{k_A+1}_{n_A}\widehat \otimes \Delta^{k_B+1}_{n_B}$ be a collection of no-signalling conditional probabilities.  Then $\gamma$ corresponds to a classical causal channel $T_\gamma:\Delta_{k_{AB}}\to 
\Delta_{n_{AB}}$ (cf. the notation in Example \ref{ex:composite_class}), given by $T_\gamma(j_A,j_B|i_A,i_B)=(\mea^{i_A}_{j_A}\otimes \mea^{i_B}_{j_B})(\gamma)$. Put 
$\Phi:=\Phi_{T_\gamma}$ as in \eqref{eq:CC_T}, choosing product bases in $\Ha_A\otimes \Ha_B$ and $\Ha_{A'}\otimes \Ha_{B'}$.
Then $\Phi=(S_A\otimes S_B)(\gamma)\in \Ce_{A,A'}\widehat\otimes \Ce_{B,B'}$, but since $\Phi$ is completely positive, 
 we have $\Phi\in \Ce^{caus}_{AB,A'B'}=\Ce_{A,A'}\widetilde\otimes \Ce_{B,B'}$. Clearly, 
 \[
(R_A\otimes R_B)(\Phi)=(R_A\otimes R_B)(S_A\otimes S_B)(\gamma)=\gamma.
 \]
 We have proved the following result (cf. \cite{hosa2017achannel}).

\begin{thm}  For any collection of no-signalling conditional probabilities 
$\gamma\in \Delta^{k_A+1}_{n_A}\widehat \otimes \Delta^{k_B+1}_{n_B}$, there is a causal bipartite quantum channel 
$\Phi\in \Ce^{caus}_{AB,A'B'}$, with $d_A= k_A+1$, $d_{A'}=n_A+1$ and $d_B= k_B+1$, $d_{B'}=n_B+1$, and collections of measurements 
$R_A\in \aff(\Ce_{A,A'},\Delta^{k_A+1}_{n_A})$ and $R_B\in \aff(\Ce_{B,B'},\Delta^{k_B+1}_{n_B})$ such that $(R_A\otimes R_B)(\Phi)=\gamma$.

\end{thm}

\begin{ex} Let $\Simp_A=\Simp_B=\square_2$ and 
choose 
\[
\gamma=\frac12((\simp_{0,0}-\simp_{1,0})\otimes \simp_{0,0}+\simp_{1,1}\otimes\simp_{1,0}+\simp_{1,0}\otimes \simp_{01})\in \square_2\widehat\otimes\square_2.
\]
We obtain the same c-c bipartite channel and sets of measurements attaining maximal CHSH violation as in \cite{bgnp2001causal,plzi2017popescu}.
\end{ex}

\section{Conclusion and further questions}

We have studied incompatibility of measurements in a family of convex finite dimensional GPTs by representing collections of measurements as affine maps 
 into a polysimplex. We have shown how properties of these maps (like being ETB or a retraction) are tied to incompatibility. We introduced incompatibility witnesses and used them to characterize incompatibility degree. Our results suggest that incompatibility 
 is closely related to the geometry of polysimplices, for example, the maximal incompatibility degree attainable for a given state space $K$ can be obtained by considering positive maps on $K$ that factorize through the polysimplex of the given shape. 

Our setting allows us to study the relations of incompatibility to steering and non-locality through incompatibility witnesses. We have shown that the Bell scenario provides incompatibility witnesses of a restricted type, more precisely, factorizing through the dual of a polysimplex. In general, the family of such witnesses is strictly smaller than the set of all witnesses and is therefore weaker for the detection of incompatibility. This explains  the existence of incompatible collections of measurements that do not violate Bell inequalities and suggests that this feature is not specific for quantum theory, but is common in GPT.

 There is a number of questions left for further research. For example, the incompatibility degree attainable by more general collections of quantum measurement can be investigated using witnesses as in Corollary \ref{coro:quantummax}. For this, a characterization of  extremal maps of a polysimplex into $B(\Ha)^+$  would be useful. For the study of relations between incompatibility and non-locality,
 one could describe the  witnesses that can be obtained from Bell inequalities as in Theorem \ref{thm:bell_id}. It is an interesting question to what extend are 
 these witnesses weaker and how it depends on the theory in question. 
It might be also worthwhile to study collections of quantum measurements and their incompatibility, as well as steering and Bell nonlocality,  within the framework of quantum networks as suggested in Remark \ref{rem:incomp_in_combs}. 

Maps into a polysimplex can be used to  describe collections of measurements only up to a fixed size and number of outcomes. 
The isomorphism with (faces of) classical channels (Proposition \ref{prop:ccchannel_polysimp}) suggests that it might be possible 
 do include all collections of measurements as affine maps into the space of more general Markov kernels.
 
\appendix
\renewcommand{\thesection}{\Alph{section}}
\setcounter{equation}{0}
\renewcommand{\theequation}{\thesection.\arabic{equation}}
\setcounter{prop}{0}
\setcounter{lemma}{0}
\renewcommand{\theprop}{\thesection.\arabic{prop}}
\renewcommand{\thelemma}{\thesection.\arabic{lemma}}
\section{Positive and ETB maps}\label{app:positive}

We list some well known results on the cones of positive maps.
Let $K$ be a compact convex set. Let $x_0,\dots,x_n$ be a basis of $V(K)$ and let $e_0,\dots,e_n\in A(K)$ be the dual basis.  Put 
\begin{equation}\label{eq:chik}
\chi_K:=\sum_i x_i\otimes e_i\in V(K)\otimes A(K).
\end{equation}
Let $f\in A(K)$, then $f=\sum_i \<f,x_i\>e_i$, therefore we have
\begin{equation}\label{eq:chik_fy}
\<\chi_K,f\otimes y\>=\sum_i \<f,x_i\>e_i(y)=f(y),\qquad \forall y\in K.
\end{equation}
Consequently, $\chi_K$ does not depend on the choice of the basis.

\begin{lemma}\label{lemma:chiKK}  Let $x'\in ri(K)$ and put $K^\dagger:=\{ f\in A(K)^+, f(x')=1\}$. 
Then 
\begin{enumerate}
\item[(i)] $K^\dagger$ is a compact convex set  and we have $V(K^\dagger)^+=A(K)^+$, $A(K^\dagger)^+=V(K)^+$ and $1_{K\dagger}=x'$.
\item[(ii)] $\chi_K\in K\widehat \otimes K^\dagger$.
\item[(iii)] Let $T\in \aff(K, V(K'))$, then $(T\otimes id)(\chi_K)=(id\otimes T^*)(\chi_{K'})$. 
\item[(iv)] For any $\xi\in V(K'\widehat\otimes K^\dagger)^+$, 
there is a unique $T\in \aff(K,V(K')^+)$ such that  $(T\otimes id)(\chi_K)=\xi$, determined by
\begin{equation}\label{eq:vector_repT}
\<T(x),f'\>=\<\xi, f'\otimes x\>,\qquad \forall x\in K, f'\in A(K').
\end{equation}

\end{enumerate}

\end{lemma}

\begin{proof} Since $x'\in ri(K)$, $f(x')=0$ for $f\in A(K)^+$ implies that $f=0$, so that $K^\dagger$ is a base of $A(K)^+$. 
This proves (i).  The statement (ii) follows immediately from \eqref{eq:chik_fy} and (i). For (iii), let $f'\in A(K')$ and $y\in K$, then by \eqref{eq:chik_fy}, we have
\begin{align*}
\<(T\otimes id)(\chi_K),f'\otimes y\>&=\<\chi_K,T^*(f')\otimes y\>=T^*(f')(y)\\
&=\<f',T(y)\>= \<\chi_{K'}, f'\otimes T(y)\>\\
&=\<(id\otimes T^*)(\chi_{K'}),f'\otimes y\>.
\end{align*}
For (iv), it is clear that \eqref{eq:vector_repT} determines an element $T\in \aff(K,V(K')^+)$ and $(T\otimes id)(\chi_K)=\xi$ holds by \eqref{eq:chik_fy}.

\end{proof}

We now have the following characterizations of ETB maps.

\begin{prop}\label{prop:ETBapp} Let $T\in \aff(K,V(K')^+)$. The following are equivalent.
\begin{enumerate}
\item[(i)] $T$ is ETB.
\item[(ii)] $(T\otimes id)(\chi_K)$ is separable.
\item[(iii)] $T$  factorizes through a simplex: there are a simplex $\Delta_n$ and maps $T_0\in \aff(K,V(\Delta_n)^+)$ and $T_1\in \aff(\Delta_n, V(K')^+)$ such that $T=T_1T_0$. 
\end{enumerate}
If $T$ is a channel, $T_0$ and $T_1$ in (iii)  may be chosen to be channels as well.
\end{prop}

\begin{proof} (i) $\implies$ (ii) is clear. Assume (ii), then there are some $\phi_j\in V(K')^+$ and $f_j\in A(K)^+$ such that 
$(T\otimes id)(\chi_{K})=\sum_{j=0}^n \phi_j\otimes f_j$. By \eqref{eq:vector_repT}, we have for $y\in K$, $g\in A(K')$,
\[
\<T(y),g\>=\<\sum_j\phi_j\otimes f_j,g\otimes y\>=\<\sum_j f_j(y)\phi_j,g\>, 
\]
so that $T=\sum_j f_j(\cdot)\phi_j$. Let $\delta_j$ be the extreme points of $\Delta_n$ and put $T_0=\sum_jf_j(\cdot)\delta_j $,
 $T_1(\delta_j)=\phi_j$, then $T=T_1T_0$. (iii) $\implies$ (i) since any $T_0\in \aff(K,V(\Delta_n)^+)$ is ETB, see Example \ref{ex:composite_class}.

Assume that $T$ is an ETB channel, so that $T=\sum_j f_j(\cdot)\phi_j$ as above.  Let 
$c_j:=\<\phi_j,1_{K'}\>$. We may assume $c_j>0$, otherwise we may replace $\Delta_n$ by a smaller simplex. 
Put  $\tilde f_j:=c_jf_j$ and $\tilde \phi_j:=c_j^{-1}\phi_j$, then the corresponding maps $\tilde T_0$ and $\tilde T_1$ 
 are channels such that $T=\tilde T_1 \tilde T_0$.
 
 \end{proof}

We next prove Proposition \ref{prop:aff_dual}. Let $x_0,\dots,x_n\in K$ be a basis of $V(K)$ and $e_0,\dots,e_n$ the dual basis, then 
\[
\Tr T=\sum_i \<e_i,T(x_i)\> =\<(T\otimes id)\chi_K,\chi_{K^\dagger}\>.
\]

\begin{proof}[Proof of Proposition \ref{prop:aff_dual}]  Let $S\in \aff_{sep}(K',V(K))$, so that $S=\sum_j f'_j(\cdot)\phi_j$ for some $f'_j\in A(K')^+$ and $\phi_j\in V(K)^+$.
Then $\Tr TS= \sum_j \<T(\phi_j),f'_j\>\ge 0$. Conversely, assume that  for all $T\in \aff(K,V(K')^+)$
\[
\Tr ST=\<(T\otimes id)\chi_{K},(S^*\otimes id)\chi_{K^\dagger}\>\ge 0.
\]
By Lemma \ref{lemma:chiKK} (iv), we have $(S^*\otimes id)\chi_{K^\dagger}\in A(K\widehat\otimes K^\dagger)^+$, so it must be separable. By Proposition \ref{prop:ETBapp}, $S^*$, and hence also $S$, is ETB. 

\end{proof}

\setcounter{prop}{0}
\section{The cones $\aff(\Simp, V(K)^+)$ and $\aff_{sep}(\Simp, V(K)^+)$  }\label{app:cones}

We describe the cone of positive maps $\aff(\Simp, V(K)^+)$ and characterize the ETB ones.

\begin{prop}\label{prop:witness} The elements  $w_{n_0,\dots,n_k}\in V(K)^+$, $n_i=0,\dots,l_i$, $i=0,\dots,k$
 are vertices of some $W\in\aff(\Simp,V(K)^+)$ if and only if they satisfy
\begin{align}
w_{n_0,\dots,n_k}+w_{n_0',\dots,n_k'}&=w_{n_0,\dots,n_{i-1},n_i',n_{i+1},\dots,n_k}\notag\\
&+w_{n_0',\dots, n_{i-1}',n_i,n_{i+1}',\dots,n_k'}\label{eq:poly}
\end{align}
for all $n_0,\dots,n_k$, $n_0',\dots,n_k'$ and $i$.
Moreover, $W$ is ETB if and only if 
 there are some $\psi^i_j\in V(K)^+$, $j=0,\dots,l_i$, $i=0,\dots,k$,  such that
\[
w_{n_0,\dots,n_k}=\sum_{i=0}^k\psi^i_{n_i}.
\]

\end{prop}

\begin{proof} Let $w_{n_),\dots,n_k}$ be vertices of $W\in\aff(\Simp,V(K)^+)$. We have 
\begin{align*}
\frac12&(\simp_{n_0,\dots,n_k}+\simp_{n_0',\dots,n_k'})
\\&=\left(\frac12(\delta_{n_0}+\delta_{n_0'}),\dots,\frac12(\delta_{n_k}+\delta_{n_k'})\right)\\
 &=\left(\frac12(\delta_{n_0}+\delta_{n_0'}),\dots,\frac12(\delta_{n_i'}+\delta_{n_i}),\dots,\frac12(\delta_{n_k}+\delta_{n_k'})\right)\\
&=\frac12(\simp_{n_0,\dots,n_{i-1},n_i',n_{i+1},\dots,n_k}+\simp_{n_0',\dots, n_{i-1}',n_i,n_{i+1}',\dots,n_k'}),
\end{align*}
hence \eqref{eq:poly} must hold. 
 Conversely, assume $w_{n_0,\dots,n_k}$ satisfy  \eqref{eq:poly} and put
\begin{align*}
W(\simp_{l_0,\dots,l_k})&:=w_{l_0,\dots,l_k},\\ 
W(\edg^i_j)&:=w_{l_0,\dots,l_{i-1},j,l_{i+1},\dots,l_k}-w_{l_0,\dots,l_k}.
\end{align*}
This determines a map  $W\in \aff(\Simp,V(K))$. By \eqref{eq:vertices}, we have $W(\simp_{n_0,\dots,n_k})=w_{l_0,\dots,l_k}+\sum_{i=0}^k W(\edg^i_{n_i})$. Using repeatedly the relations \eqref{eq:poly}, we get $W(\simp_{n_0,\dots,n_k})=w_{n_0,\dots,n_k}$. For the second statement, note that since the effects $\mea^i_j$ generate $A(\Simp)^+$, $W$ is ETB if and only if  
 there are $\psi^i_j\in V(K)^+$ such that
\begin{equation}\label{eq:etbw}
W=\sum_{i=1}^k\sum_{j=0}^{l_i} \mea^i_j(\cdot)\psi^i_j.
\end{equation}
Applying this to the vertices of $\Simp$, we obtain  the statement.

\end{proof}

\begin{acknowledgments}
This research was supported by the grants  VEGA 2/0069/16 and APVV-16-0073.
\end{acknowledgments}


%
\end{document}